\documentclass{egpubl}
\usepackage{egsgp2025}
 
\ConferencePaper        % uncomment for (final) Conference Paper
\SpecialIssuePaper         % uncomment for final version of Computer Graphics Forum, special issue
\CGFccby
\usepackage[T1]{fontenc}
\usepackage{dfadobe}  

\usepackage{cite}  % comment out for biblatex with backend=biber 
\BibtexOrBiblatex
\electronicVersion
\PrintedOrElectronic

\ifpdf \usepackage[pdftex]{graphicx} \pdfcompresslevel=9
\else \usepackage[dvips]{graphicx} \fi

\usepackage{egweblnk} 
\usepackage{amssymb}
\usepackage{booktabs} % For formal tables
\usepackage{amsmath}
\usepackage{xcolor}
\usepackage{amsfonts}
\usepackage{graphicx} 
\usepackage{wrapfig}
\usepackage{algpseudocode,algorithm,algorithmicx}
\usepackage{mathtools}

\graphicspath{ {./figures/} }

\renewcommand{\sec}[1]{Sec.~\ref{sec:#1}} %the standard \sec command is *secant*, who needs that anyway.
\newcommand{\fig}[1]{Fig.~\ref{fig:#1}}
\newtheorem{theorem}{Theorem}[section]

\newtheorem{lemma}[theorem]{Lemma}
\newcommand{\new}[1]{#1}
\newtheorem{definition}{Definition}[section]
\newtheorem{proposition}[theorem]{Proposition}
\title[One-Shot Method for Computing Generalized Winding Numbers]%
      {One-Shot Method for Computing Generalized Winding Numbers}

\author[C. Martens \& M. Bessmeltsev]
{\parbox{\textwidth}{\centering C. Martens\orcid{0009-0005-6866-8159}
        and M. Bessmeltsev\orcid{0000-0002-8864-2934}}
        \\
{\parbox{\textwidth}{\centering Université de Montréal, Canada}
}
}
\begin{document}

\teaser{
\centering
    \includegraphics[width=1\linewidth]{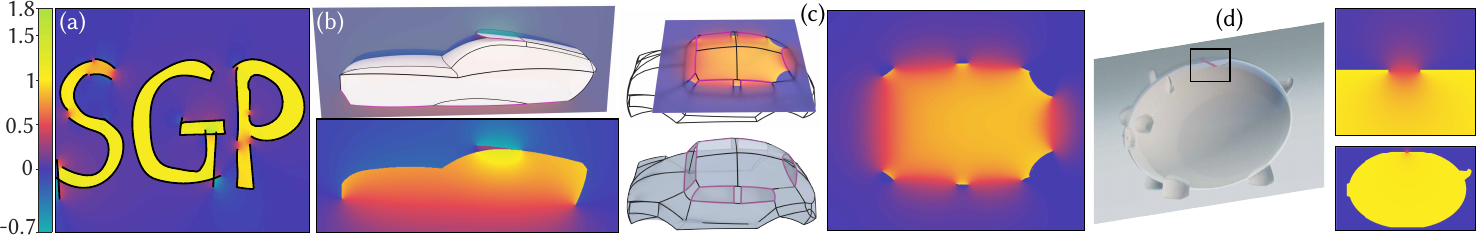}
    \vspace{-5pt}
    \caption{We propose a novel \emph{one-shot} method for generalized winding number computation that does not require discretization of the surface. Our method only uses the \emph{boundary} of the surface and a single ray-surface intersection test. This focus on the boundary allows us to compute winding numbers for 2D parametric curves (a) on a regular grid of query points significantly faster than the state of the art. In 3D, we can compute a winding number for a parametric surface (b, Coons patches, c, minimal surfaces) without discretizing it. For both parametric surfaces (b,c) and meshes (d), for some boundaries our method can improve the performance over standard methods while staying exact. Piggy Bank mesh with minor modifications by \href{https://www.thingiverse.com/thing:104734}{belch}.}
    \label{fig:teaser}
}

\maketitle
%-------------------------------------------------------------------------
\begin{abstract}
%\raggedright
   The generalized winding number is an essential part of the geometry processing toolkit, allowing to quantify how much a given point is inside a surface, even when the surface has boundaries and noise. We propose a new universal method to compute a generalized winding number, based only on the surface boundary and the intersections of a single ray with the surface, supporting any oriented surface representations that support a ray intersection query. Due to the focus on the boundary, our algorithm has a unique set of properties. For 2D parametric curves, \textit{on a regular grid of query points}, our method is up to \textbf{$4\times$} faster than the current state of the art, maintaining the same precision. In 3D, our method can compute a winding number of a surface without discretizing it, including parametric surfaces. For some meshes with many triangles and a simple boundary, our method is faster than the hierarchical evaluation of the generalized winding number while still being precise. Similarly, on some parametric surfaces with a simple boundary, our method can be faster than adaptive quadrature. We validate our algorithms theoretically, numerically, and by demonstrating a gallery of results on a variety of parametric surfaces and meshes, as well uses in a variety of applications, including voxelizations and boolean operations.
%-------------------------------------------------------------------------
%  ACM CCS 1998
%  (see https://www.acm.org/publications/computing-classification-system/1998)
% \begin{classification} % according to https://www.acm.org/publications/computing-classification-system/1998
% \CCScat{Computer Graphics}{I.3.3}{Picture/Image Generation}{Line and curve generation}
% \end{classification}
%-------------------------------------------------------------------------
%  ACM CCS 2012
%   (see https://www.acm.org/publications/class-2012)
%The tool at \url{http://dl.acm.org/ccs.cfm} can be used to generate
% CCS codes.
%Example:


\end{abstract}  

\section{Introduction}
For meshes, parametric surfaces, and point clouds, even representing open, noisy, and non-manifold geometry, the notion of \emph{generalized winding number} captures \emph{how much} a point is inside the surface \cite{Jacobson-13-winding,barillFastWindingNumbers2018} (Fig.~\ref{fig:parametric_gaps}). Generalized winding numbers are used in a variety of applications, including tetrahedral meshing \cite{tetMeshingWild}, surface reconstruction \cite{xuGloballyConsistentNormal2023}, and medical imaging \cite{becciu3DWindingNumber2011}, to name a few. 

The standard methods of computing generalized winding numbers calculate the area of the surface projected onto a unit sphere. %For parametric surfaces, which often contain gaps (Fig.~\ref{fig:parametric_gaps}a,b) and thus also may benefit from generalized winding number, computing such projected area requires surface discretization. This discretization may defeat the purpose of having a high-precision parametric surface, as it can lead to a significant loss of precision (Fig.~\ref{fig:parametric_gaps}c,d). 
We propose an alternative method to compute the generalized winding number. Our main insight is that for any surface, with or without boundary, computing a generalized winding number at a point requires the surface \emph{boundary} only and intersections of the surface with a \emph{single} ray. We furthermore show that in typical scenarios, when processing multiple query points, one can use less than one ray per query point on average, significantly reducing time complexity.

Our method is generic and can be applied to various geometry representations, in 2D or 3D. The elegance lies in its universality: just like a standard ray tracing pipeline, our method is exactly the same for different surface representations, such as meshes or parametric surfaces, and encapsulates the representation completely within a ray-surface intersection test. Due to our method's focus on the boundary, the performance of our method does not depend much on the complexity of the surface and instead excels when the boundary is simple. For example in 2D, when the boundary of a set of parametric curves is just a set of points, our method, when run on a regular grid of query points is up to 4x faster than the state of the art. Similarly, for challenging parametric surfaces with a planar boundary (Fig.~\ref{fig:surf_revolution}), on a similar arrangement of query points, our method is significantly faster than the standard integration method, adaptive quadrature, while maintaining the same precision. For meshes with few boundary edges, in typical application scenarios such as voxelization, our method can be significantly faster than the hierarchical evaluation of the winding number \cite{Jacobson-13-winding}. While an approximation can be computed more efficiently \cite{barillFastWindingNumbers2018}, our computation is significantly more precise. The regular grid of query points allows us to reuse the same ray intersection for many points, saving computations. We additionally show that the key components of our method, including the ray-surface intersection, can be done exactly via a convex Sum-of-Squares formulation \cite{SOS}, making the computations exact if so desired.

We showcase our method on a wide variety of standard tasks for generalized winding numbers for all these representations (Sec.~\ref{sec:results}), in 2D and 3D, including parametric Coons patches and B\'ezier curves and triangles. We extensively validate our method by both formally showing it is equal to the generalized winding number and validating its numerical precision and performance in numerous experiments, as compared to the state of the art methods including \cite{Jacobson-13-winding,barillFastWindingNumbers2018,RobustContainmentQueries2024a} and adaptive quadrature (Sec.~\ref{sec:results}). An application of our method is to piecewise minimal surfaces on 3D curve networks, collections of 3D curves forming one or multiple loops (Fig.~\ref{fig:teaser}), typical for (VR/AR) contexts \cite{GoogleTiltBrush2023,gryaditskaya2020lifting} (Fig.~\ref{fig:motivation}). 
We apply our method to compute a winding number of piecewise minimal surfaces by expressing them in a parametric form via Boundary Element Method (BEM). Our method allows us to avoid discretizing these surfaces.

Our contributions are:

\begin{itemize}
	\item demonstrating that a computation of generalized winding number requires boundary-only operations and finding intersections of a surface with a \emph{single} ray and that the same ray can be reused for multiple query points,
	\item a new method of computing generalized winding number based on this observation, enabling state-of-the-art performance for 2D parametric geometry on a regular grid of query points, and showing performance improvements for some 3D parametric surfaces and meshes.
\end{itemize}

\begin{figure}
	\includegraphics[width=\linewidth]{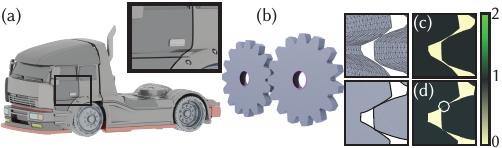}
	\caption{Just like meshes, parametric surfaces also often contain gaps and overlaps, either intentional for tight alignment after manufacturing (a) or unintentional due to modelling defects. At the same time, precision is paramount for parametric surfaces, especially when parts are designed to touch (b). Discretizing such surfaces is often undesired, as it may introduce new artifacts (c). Our method can compute winding numbers precisely without surface discretization (d). \href{https://grabcad.com/library/trailer-head-1}{Truck} by Toru Odazawa.}
	\label{fig:parametric_gaps}
\end{figure}

\section{Related Work}

\textbf{Winding Numbers.} The classical concept of a winding number, originally defined for a closed curve \cite{meisterGeneraliaGenesiFigurarum1769} and easily generalizable to closed surfaces, is one of the topics in complex analysis \cite{needhamVisualComplexAnalysis1997}. \cite{Jacobson-13-winding} generalized winding numbers to polylines and meshes, including the ones with boundary. They also propose an efficient algorithm to compute the generalized winding number that empirically has a time complexity of $O(F^{0.55})$, where $F$ is the number of triangles. Later, winding numbers have been further generalized to, among other representations, point clouds and triangle soups \cite{barillFastWindingNumbers2018}. The method of \cite{barillFastWindingNumbers2018} can also serve as a very efficient approximation for meshes, achieving $O(\log F)$. Under simple assumptions, on meshes our algorithm has a time complexity between those two, while still being exact (Sec.~\ref{sec:validation}). More importantly, these methods target discretized geometry; applying them to parametric surfaces leads to loss of precision (Fig.~\ref{fig:accuracy_err}). Our method can be applied to parametric surfaces directly yielding exact winding numbers. Finally, \cite{fengWindingNumbersDiscrete2023} extended the definition to curves on discrete surfaces.

Our method computes a generalized winding number of a surface in any representation that supports ray intersection queries, such as parametric surfaces or meshes. A recent work \cite{RobustContainmentQueries2024a} offers a practical method of finding a generalized winding number for a rational curve in 2D; their technique cannot be easily extended into 3D. In contrast, we propose a general method for any surfaces or curves supporting ray intersections, including parametric or discrete geometry in 2D and 3D. We compare with their method in 2D in Sec.~\ref{sec:results}.

%\setlength{\columnsep}{3pt}
%\begin{wrapfigure}{r}{0.2\linewidth}
%	\vspace{-10pt}
%	\begin{center}
%			\includegraphics[width=0.7in]{figures/S1_S2}
%	\end{center}
%	\vspace{-10pt}
%\end{wrapfigure}
A closely related concept is a solid angle, which for a surface is defined as generalized winding number (without normalization) modulo $4\pi$, so it does not measure the number of turns the surface makes around a point, nor does it change under a flip of surface orientation \cite{PerspectivesWindingNumbers2023}. We are focusing on the full winding number, the signed solid angle, that includes the multiplicity and depends on the orientation. The solid angle can be similarly defined and effectively computed for a space \emph{curve} \cite{solid_angle_curve,chern2024area,binyshMaxwellTheorySolid2018}. Note that in 3D we always compute a winding number of a \emph{surface}, albeit sometimes not discretized. \setlength{\columnsep}{1em}

\textbf{Ray Casting.} 
Our work is influenced by works on ray casting to voxelize shapes \cite{nooruddinSimplificationRepairPolygonal2003,houstonUnifiedApproachModeling}. The surfaces they focus on have no true boundary, albeit might be composed of multiple components due to noise, so they perform an inside-outside test; ours may have a boundary, so the definition of insidedness is not binary and is better captured by the notion of generalized winding number. %An interesting connection is the work of \cite{hampelCouplingBoundaryElements2008} that uses BEM to perform ray tracing. 

\section{Background}
\label{sec:bg}

We start by defining a generalized winding number for a smooth oriented curve $C$ in 2D and an oriented surface $M$ in 3D, open or closed. \new{Both the theory and the method we propose apply equally to piecewise smooth curves or surfaces. We only use smoothness for the simplicity of exposition.} For brevity, later in the text we write `winding number' instead of `generalized winding number'. 

\begin{figure}
	\centering
	\includegraphics[width=0.6\linewidth]{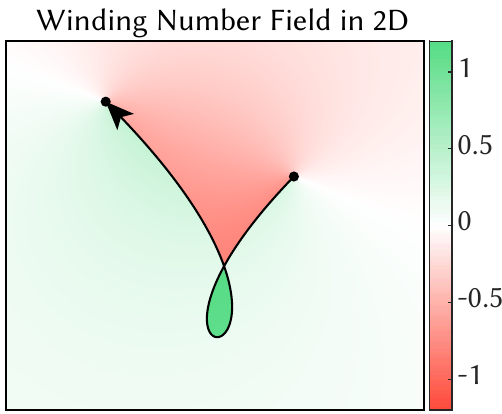}
	\caption{In two dimensions, the winding number is the sum of the signed subtended angles of the curve.}
	\label{fig:wn_2d}
\end{figure}
%\cedric{discuss 2d boundary point orientation } %why here? there's nothing here about the boundary
%\cedric{clarify what we mean by exactness. Note that Spainhour also makes an exactness claim for his 2d method}
%not here

In 2D, a winding number at a point $p \in \mathbb{R}^2 \setminus C$ is defined via the subtended angle $\theta$ of the curve $C$:
\[
 w_C(p) = \frac{1}{2\pi}\int_C{d\theta(p)}.
\]

In 3D, for a given surface $M$, the winding number at a point $p\in\mathbb{R}^3~\setminus~M$ is defined as the integral of the differential solid angle $d\Omega(p)$ subtended at $p$ \cite{barillFastWindingNumbers2018}:
\[
w_M(p) = \frac{1}{4\pi}\int_M{d\Omega(p)}.
\]

In the article we focus on the 3D case for brevity; our algorithm works in both 2D and 3D, as we demonstrate in Fig.~\ref{fig:teaser} and Sec.~\ref{sec:results}. In both the theory and the algorithm we consider the winding number to be undefined when $p \in C$ or $p \in M$.

%\cedric{The winding number is related to the signed solid angle.} %why do we care?
%\begin{equation}
%	\Omega(p) = \iint_S \frac{\hat{r} \cdot \hat{n}}{r^2} dS = 4 \pi w_M(p)
%	\label{eq:signed_solid_angle}
%\end{equation}
%\cedric{Where $\hat{r}$ is the difference between $p$ and $q$, where as $r$ is euclidean distance between both vectors. This definition tied to the solid angle is useful for calculating the winding number on parametric surfaces.}

 An equivalent formulation of the winding number, which we use as the starting point for our method, is the number of signed intersections of rays with the surface over all directions \cite{Jacobson-13-winding}:
\begin{equation}
	w_M(p) = \frac{1}{4\pi}\int_{q \in S^2_p} \chi(q) dA,
	\label{eq:wn_chi}
\end{equation}
where $S^2_p$ is a unit sphere 
centered at $p$, and $\chi(q)$ is the number of signed intersections between ray $\overrightarrow{pq}$ and $M$. By `signed intersections' we mean that an intersection is counted as $+1$ when the ray is intersecting from the back of the surface, i.e., the ray is aligned with the surface normal, and as $-1$ when it is intersecting from the front. So formally,
\begin{equation}
	\chi(q) = \sum_j \mathrm{Sign}(\overrightarrow{pq} \cdot n(r_j)),
	\label{eq:chi}
\end{equation}
where $r_j$ are all the intersection points of a ray $\overrightarrow{pq}$ with the surface, and $n(r_j)$ are the corresponding normals.  We refer to $\chi$ as the number of signed intersections.
\section{Algorithm}
\label{sec:algo}

The input to our algorithm is an oriented manifold surface $M$ \new{of arbitrary genus (Fig.~\ref{fig:genus})}, with or without boundary $\partial M$. We first describe a generic algorithm, then clarify its implementation for parametric and discrete surfaces in Sec.~\ref{sec:impl_details}. 
\begin{figure}
	\includegraphics[width=\linewidth]{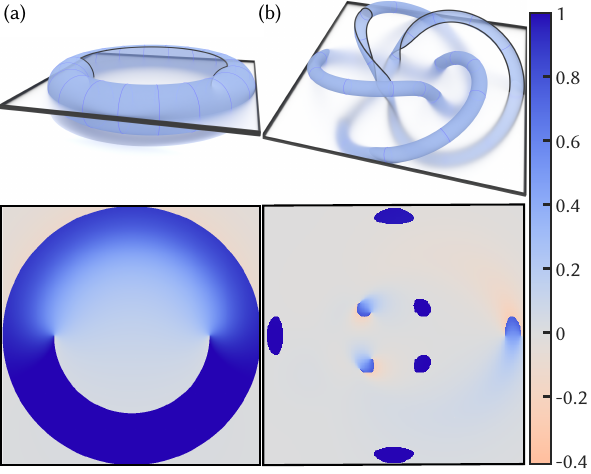}
	\caption{\new{Our method is compatible with oriented surfaces of arbitrary genus. We show a torus (a) with genus 1, and a (4,3)-torus knot with genus 3 (b). Surface boundaries are highlighted in black.}}
	\label{fig:genus}
\end{figure}

In this section we outline our basic algorithm that uses multiple rays for a query point. Later (Sec.~\ref{sec:one_shot}), we build upon it and describe a significantly faster \emph{one-shot} algorithm, which uses a single ray for multiple query points. 

\subsection{Computing the Winding Number}
Our goal is to compute the winding number at a point $p \in \mathbb{R}^3$, i.e., evaluate the integral in Eq.~\ref{eq:wn_chi} using the boundary $\partial M$. Computing $\chi(q)$ for a single $q$ includes finding all intersections of a ray $\overrightarrow{pq}$ with a potentially non-discretized surface --- a rather expensive operation (\sec{ray_int}). Therefore, evaluating the integral na\"ively is infeasible, as it requires many such ray-surface intersections. Instead, our main observation is the following lemma, a corollary of a more general and technical proposition we prove in \new{Appendix \ref{app:lemma_proof}}:

\begin{lemma}
 Let $M \subset \mathbb{R}^3$ be a smooth surface with a boundary $\partial M$. Let $S^2_p$ be a unit sphere centered at  $p \in \mathbb{R}^3$, and let the projection of $\partial M$ onto $S^2_p$ split the sphere into open regions $A_i$. Then for each $i$ there is a constant $\chi_i \in \mathbb{Z}$ such that $\chi(q)=\chi_i$ almost everywhere in each $A_i \ni q$. 
 \label{lem:main}
 \end{lemma}
 
 \begin{figure}
 	\includegraphics[width=\linewidth]{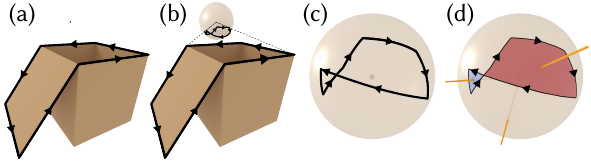}
 	\caption{Overview of the basic algorithm: Given a surface with a boundary (a), we project the boundary onto a unit sphere around the query point (b, blowup in c).  The boundary splits the surface of the sphere into regions (d). We shoot a single ray through each of those regions, and compute the number of signed intersections $\chi_i$ of each ray with the surface.}
 	\label{fig:overview}
 \end{figure}
 
\begin{figure}
	\includegraphics[width=\linewidth]{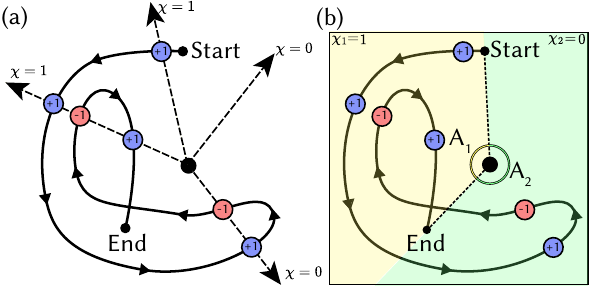}
	  \caption{We define $\chi$ for a ray as the number of signed intersections (Eq. \ref{eq:chi}) (a). The value of $\chi$ is constant in each region, $A_1$ and $A_2$, formed by projecting the curve endpoints onto a unit circle. Note that this holds even if the curve is self-intersecting.}
	\label{fig:signed_intersection}
\end{figure}

The lemma is illustrated in \fig{signed_intersection}. Intuitively, this means that $\chi$ is constant within each region, except maybe for a few isolated curves or points. In the figure, we show the 2D example: the surface $M$ becomes a curve, the boundary of $M$ is now a pair of endpoints; their projections onto the unit circle split it into two regions $A_1$ and $A_2$. The lemma states that even though the number of \emph{unsigned} intersections may change over each region $A_i$ on the sphere, the number of \emph{signed} intersections is a constant within each region (except for some special points, see below). We denote that constant as $\chi_i$. The projection of the boundary itself has measure zero, thus we do not need to define $\chi(q), q \in \textrm{Proj}_{S_p^2}\partial M$, as these values do not influence the integral. The lemma therefore allows us to rewrite Eq.~\ref{eq:wn_chi} as a discrete sum:

\begin{equation}
 w_M(p) = \frac{1}{4\pi} \sum_i \textrm{Area}_{S^2_p}(A_i) \chi_i.
 \label{eq:wn_discrete_sum}
\end{equation}

\setlength{\columnsep}{1em}
\begin{wrapfigure}[5]{r}{0.7in}
	\centering
	\vspace{-\intextsep}
	\includegraphics[width=0.7in]{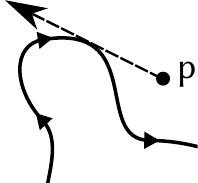}
\end{wrapfigure}

The technical detail \new{'}almost everywhere' in the formulation of the lemma refers to the points $q$ where the rays $\overrightarrow{pq}$ are tangent to the surface at a contact point (see inset). The set of such points, however, has measure zero, so they do not affect the integral.

The observation captured in this lemma significantly simplifies our algorithm. Now we can shoot one ray per region $A_i$, compute the number of signed intersections $\chi_i$, and sum their contributions multiplied by the region area $\textrm{Area}_{S^2_p}(A_i)$. Considering that number of regions is typically small, this leads to an efficient algorithm: project the boundary onto the sphere (\fig{overview}c), decompose the sphere\new{'s} surface into regions, compute their areas, shoot a single ray through each one and compute $\chi_i$ (\fig{overview}d). \new{Therefore, our method is compatible with all oriented surface representations that support ray-intersection queries and represent boundaries as continuous curves. Examples of unsupported representations include oriented point clouds and NeRFs \cite{NeRF21}, as they neither provide explicit ray-intersections nor represent boundaries as continuous curves.}

\begin{figure}
	\centering
	\includegraphics[width=0.8\linewidth]{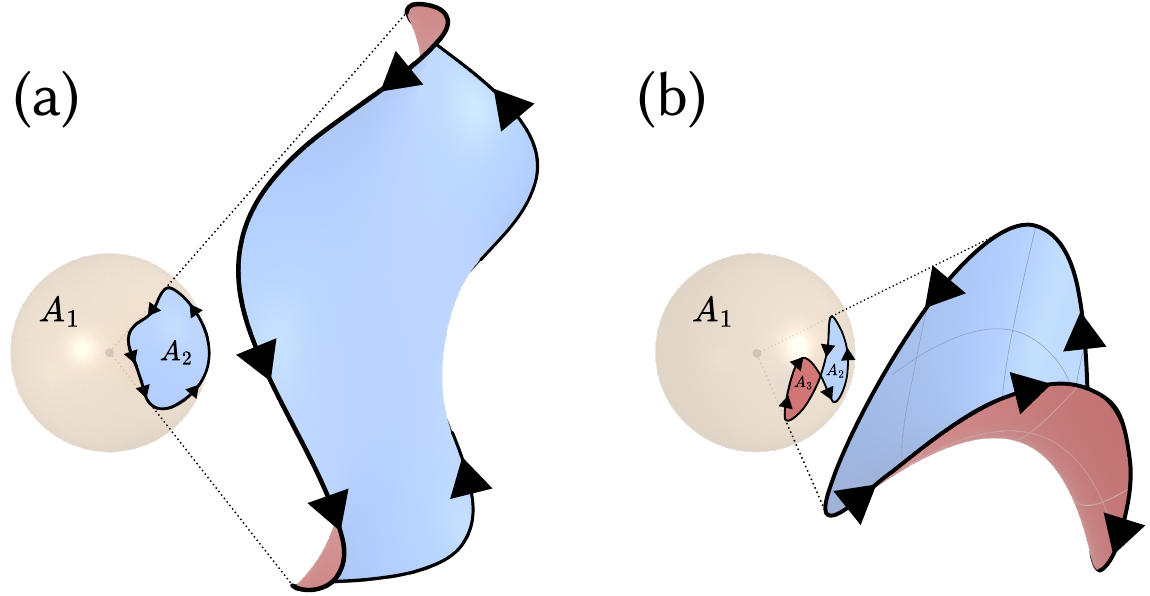}
	\caption{For a given query point and a unit sphere around it, the projection of the boundary of a surface patch splits the sphere into two ((a), the blue `inside' and the beige `outside') or more regions (b) $A_i$ with a constant number of signed intersections $\chi_i$. Blue is the `front' of the surface, i.e., the side with the positive normal; the boundary is oriented accordingly.}
	\label{fig:sphere_proj}
\end{figure}

%\subsubsection{Decomposition into Regions}
%\label{sec:decomposition}
If the curve network has no boundary, i.e., represents a closed manifold, then the sphere only has one region $A_1$, so we compute only one $\chi_1$. Such manifolds divide the space into `inside' and `outside' (see Jordan–Brouwer separation theorem \cite{guilleminDifferentialTopology2010}), so in this case our algorithm is equivalent to a simple inside-outside ray casting test. 

For a compact manifold \emph{with} boundary, the projection of the boundary $\textrm{Proj}_{S^2_p}\partial M$ on the unit sphere $S^2_p$ divides the surface of the sphere into two or more regions $A_i$. For example, for a surface with a single boundary, $\textrm{Proj}_{S^2_p}\partial M$ divides the sphere into exactly two regions if the boundary projection has no self-intersections, and more if it has (\fig{sphere_proj}).

%\begin{figure}
%\includegraphics[width=0.5\linewidth]{figures/rays_3d}
%\caption{Shooting rays through each region of the unit sphere.}
%\label{fig:shooting_rays}
%\end{figure}
\section{One-shot algorithm}
\label{sec:one_shot}
Performing ray-surface intersection tests is still rather expensive, so we would like to minimize their number. Our key observation enabling an efficient algorithm is the following lemma: 
\begin{lemma}
	For two regions $i, j$ adjacent across a common edge, $|\chi_i - \chi_j| = 1$.
\end{lemma}

\begin{figure}
	\includegraphics[width=\linewidth]{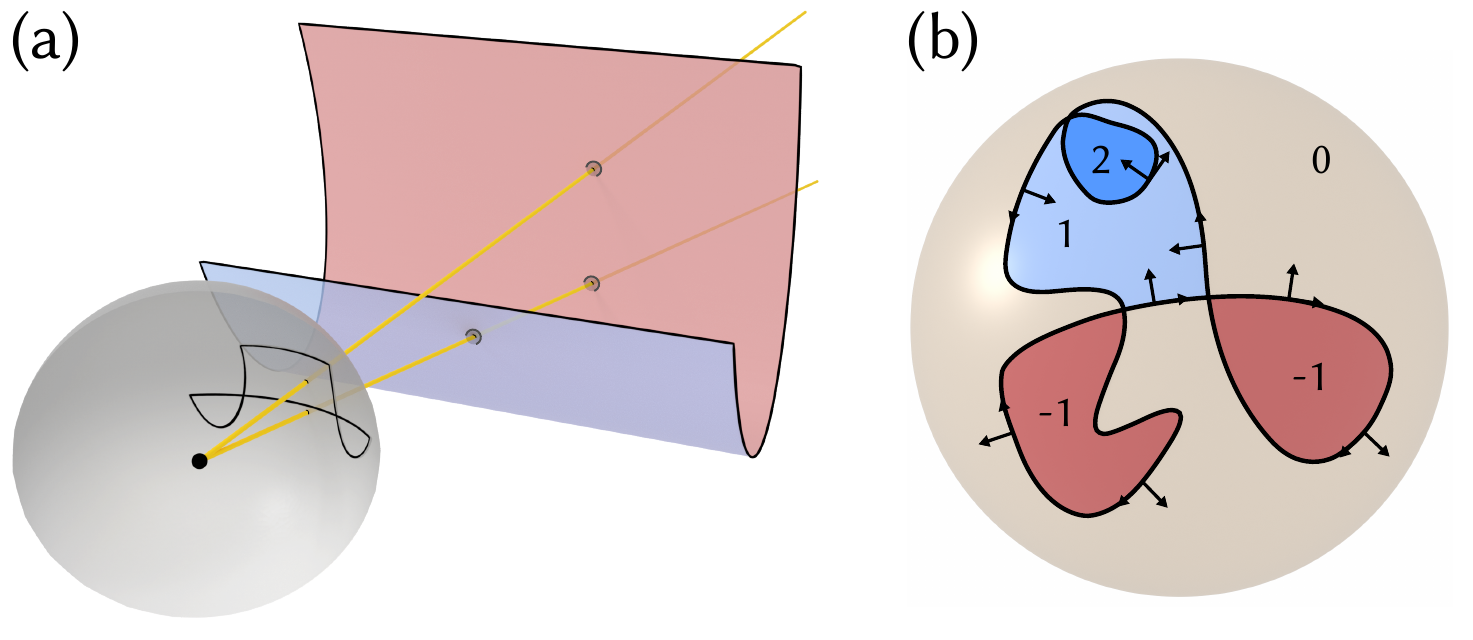}
	\caption{(a) Between two adjacent regions on the sphere, $\chi$ can only change by +1 or -1. In this example, as the ray moves from the bottom ray to the top ray, passing from one region on the sphere to the other, the number of unsigned intersections changes from 2 to 1, and $\chi$ changes from 0 to 1. (b) An example of regions and their $\chi$ values: For each region, $\chi$ is a constant. Furthermore, any two regions $i,j$ adjacent across a segment have $|\chi_i-\chi_j| = 1$. Finally, for a correctly oriented boundary curve, the normals always point towards increasing $\chi$.}
	\label{fig:relative_winding_number}
\end{figure}

Indeed, crossing the projection of the boundary $\partial M$ means one of the regions has an additional intersection with $M$ (Fig.~\ref{fig:relative_winding_number}). This lemma is also a direct corollary of the proposition in \new{Appendix \ref{app:lemma_proof}}. By `adjacent along a common edge' we mean that two regions touching only at a common point do not follow that property. For instance, in Fig.~\ref{fig:relative_winding_number}, while the $\chi=0$ and $1$ regions are adjacent across edges and thus follow the lemma, the $\chi=1$ and $\chi=-1$ regions only share a common point, so $| \chi_i-\chi_j | = 2$.

If we could know the sign of that difference for every pair of adjacent regions, computing all $\chi_i$ would be trivial once we know the value of $\chi$ for at least one region. To determine the sign of the difference $\chi_{ij} \vcentcolon= \chi_i - \chi_j$, we orient the in-plane normals to the projected boundary curve (\fig{relative_winding_number}) so that they always point towards a higher $\chi$ value. We observe that these normals are changing continuously along the projected curve boundary across the intersections (Fig.~\ref{fig:relative_winding_number}), and are always making a left turn with respect to the projected curve tangent. Therefore, we define these normals simply as $n_{S^2_p}(s) \times \tau(s)$.

We first need to compute $\chi_i$ for exactly one arbitrary region. We heuristically choose the largest region $i = \mathrm{argmax}_j \mathrm{Area}(A_j)$ and compute $\chi_i$ using the ray intersection procedure (Sec.~\ref{sec:ray_int}), using an arbitrary point $q$ within a region. For many points the largest region contains no intersections (e.g., Fig.~\ref{fig:relative_winding_number}), so most likely a ray intersection test will return quickly.

If we find at least one intersection, we additionally test whether the ray is tangent to the surface at one of those points. If it is, we ignore the ray, because it is a point of contact rather than an intersection, so the number of intersection points for this ray is unreliable, and choose a different random ray through the region. We test this by verifying $|\overrightarrow{pq} \cdot n |> \varepsilon$, where $n$ is the surface normal. For BEM, we use $\varepsilon = 10^{-2}$; for parametric geometry and meshes, $\varepsilon = 10^{-12}$.

Once all the pairwise differences $\chi_{ij}$ are known for all adjacent regions $i,j$, with a similar approach to \cite{zhouMeshArrangementsSolid2016}, we can represent adjacency between the regions as a directed graph, where each region becomes a vertex, and each pair of adjacent regions $i,j$ gets two edges $i\rightarrow j$ and $j \rightarrow i$, associated with $\chi_{ij}$ and $\chi_{ji}$, respectively. Since it is a connected graph, knowing $\chi_i$ for a single 'seed' vertex $i$ is enough to compute all the other $\chi_j$. We therefore use breadth-first search on the graph starting from $i$ and fill in all the other values of $\chi_j$, which completes our algorithm. Note that the typical number of regions, and therefore graph vertices, is fewer than 5, so this is essentially instantaneous.

\textbf{Optimization.} We further observe that once we know all the intersections along the ray $\overrightarrow{pq}$, we can immediately compute the winding numbers for all the points along that ray without any more intersection tests. More precisely, once we compute all the ray parameter values $t_j$ where it intersects the surface, we know $\chi$ for any point with parameter $t$ along that ray by only counting the intersections with $t_j > t$. Thus, for every point along that ray, we only need to perform the decomposition into spherical regions to compute its winding number.

This optimization allows us to very efficiently compute winding numbers in some typical scenarios. For instance, for all points in a planar slice of an object (Fig.~\ref{fig:harmonic}), where we only need to shoot $\min(W,H)$ rays for a $W$ by $H$ pixel image. Similarly, to compute voxelizations of resolution $N^3$, we only need to shoot $N^2$ rays (Fig.~\ref{fig:vox}). In those scenarios, this technique allows us to compute winding numbers using \textit{fewer than one ray} per point on average. 

\begin{algorithm}
	\caption{One-Shot Generalized Winding Number}
	\label{alg:main}
	\begin{algorithmic}[1] % The number tells where the line numbering should start
		\Function{Winding Number}{$p, C$} \Comment{$p \in \mathbb{R}^3$, $C$ is a closed curve}
		\State $C'$ = Project $\partial C$ onto $S^2_p$
		\State $\{A_n\}$ = IdentifyRegions($C'$)
		\State $i \gets \textrm{argmax} \; \textrm{Area}(\{A_n\})$
		\State $q \gets \textrm{NonTangentRay}(A_i)$
		\State $\{r_k, t_k, n_k\} \gets \textrm{Intersect}(\overrightarrow{pq})$ \Comment{$r_k \in \Omega, t_k \geq 0, n_k \in \mathbb{R}^3$}
		\State $\chi_i \gets\sum_k \textrm{Sign}(n_k \cdot \overrightarrow{pq})$  \Comment{Eq.~\ref{eq:chi}}
		\State $\{\chi_{ij}\} \gets \textrm{PairwiseDifferences}(\{A_n\})$
		\State $\{\chi_j\}=\textrm{BFS}(\{\chi_{ij}\}, i)$

		\State \textbf{return} $\frac{1}{4\pi} \sum_j \chi_j \cdot Area(A_j)$ \Comment{Eq.~\ref{eq:wn_discrete_sum}}
		\EndFunction
	\end{algorithmic}
\end{algorithm}

\subsection{Implementation Details}
\label{sec:impl_details}
To implement this generic algorithm, we need three main procedures: finding intersections between boundary curves projected onto the unit sphere, ray-surface intersection, and computing areas on the sphere.

\textbf{Meshes.} For meshes, we find ray-mesh intersections using a k-d tree; areas are computed using the spherical polygon area formula as simply $A = 2\pi - \sum_i (\alpha_i - \pi),$ where $\alpha_i$ are the interior angles. Finally, we identify regions by computing a surface arrangement on the sphere by first finding intersections between polylines on a sphere, which can be done using a Surface Sweep Algorithm. For a maximum constant number of intersection points, its time complexity is $\mathcal{O}(B \log B)$  \cite{ComputationGeometryAlgorithms}, where $B$ is the number of mesh boundary segments. All in all, the asymptotic time complexity of our algorithm is $\mathcal{O}(B \log B + \log F)$ for a single query point. For such scenarios as voxelization, however, the ray intersection is done only once for many query points, so the time complexity for meshes is $\mathcal{O}(B \log B + \frac{1}{\max(W,H)}\log F)$ per point, which for meshes with $B \ll F$ is faster than the state-of-the-art exact winding number computation with the empirical complexity of $\mathcal{O}(F^{0.55})$ \cite{Jacobson-13-winding}. Note that while \cite{barillFastWindingNumbers2018} computes a much faster approximation, our algorithm computes exact winding number.

\begin{figure}
	\includegraphics[width=\linewidth]{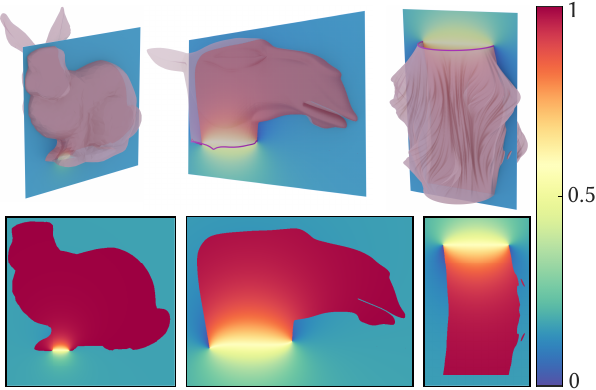}
	\caption{The results of our one-shot algorithm on meshes. Camel head mesh \cite{sorkineLaplacianMeshProcessing}, \href{https://www.thingiverse.com/thing:126567}{Vase} by virtox.}
	\label{fig:mesh}
\end{figure}

\textbf{Parametric Surfaces.} For surfaces made out of parametric patches, such as Coons patches, Bézier patches or triangles, or NURBS patches, we can compute the winding number without surface discretization as follows. Starting with the problem of ray-parametric patch intersection, \new{let us} take a B\'ezier triangle as an example, which in essence is a low-degree polynomial $f(u,v)$ with a simplex $u,v>0, u+v<1$ as the parametric domain.

Then for a given ray $O + t R$, where $O, R\in \mathbb{R}^3, t > 0$, we can formulate the following optimization problem that finds the first intersection along the ray: 
	
\begin{align}
		\begin{split}
			&\min _{u,v,t} t\\
			&\textrm{s.t.} f(u,v)  = O+tR\\
			&u, v, t\geq 0\\
			&u+v \leq 1.
		\end{split}
		\label{eq:ray_bezier_int}
\end{align}

\begin{figure}
	\centering
	\includegraphics[width=\linewidth]{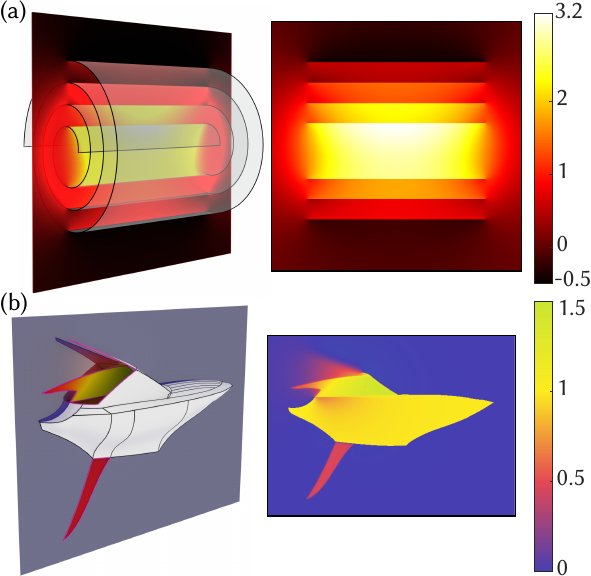}
	\caption{Additional results on parametric surfaces: a spiral with high winding numbers (a) and a non-manifold boat with Coons patches (b)}.
	\label{fig:parametric}
\end{figure}
In practice, we use a nonlinear solver to find intersections. We express the ray-surface intersection test as an equation	$f(\xi)+(q-p) t = 0$, which we solve numerically for $t \geq 0, \xi \in \Omega$, finding \emph{all} roots. We first can compute an axis-aligned bounding box of a parametric patch, then intersect it with each ray, giving $t_{min}$ and $t_{max}$; box's diagonal is $D$. If the parametric surface has no self-intersections, we then compute all roots by dividing the interval $[t_{\textrm{min}}, t_{\textrm{max}}]$ into a number of subintervals (in our implementation, we use $\max(2,30 \frac{t_{\textrm{max}}-t_{\textrm{min}}}{D})$ and initializing the solver with the middle of each subinterval. We then consider the roots duplicate, if the distance between them is less than the accuracy of the solver. \new{Note that for polynomial patches one can use specialized techniques such as `pencil of a matrix'  \cite{xiaoNoniterativeMethodRobustly2019}. We do not assume that surfaces are polynomial, so in our 3D implementation we use a generic nonlinear solver instead.}
	
We discretize the boundary of a parametric surface and compute region decomposition using polyline intersection algorithm, similar to the mesh case. Unlike \emph{surface} discretization, however, this \emph{boundary} discretization may lead only to minor differences in the final winding number close to the boundary (Sec.~\ref{sec:validation}). 

\setlength{\columnsep}{1em}
\begin{wrapfigure}[5]{r}{0.7in}
	\centering
	\vspace{-\intextsep}
	\includegraphics[width=0.7in]{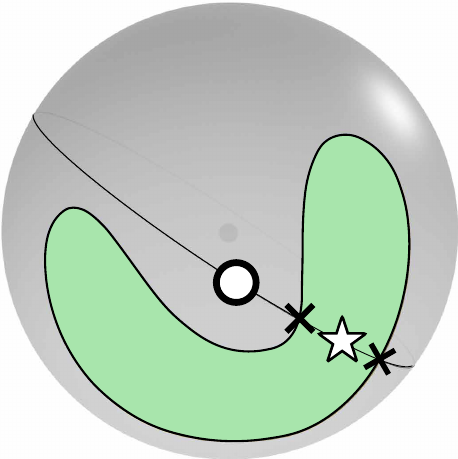}
\end{wrapfigure}
To find a point within a spherical region, we first compute the Euclidean centroid of the region,
which, due to the convexity of the sphere, lies inside the sphere. We project it onto the sphere (inset, white circle) and connect it to an arbitrary point (one of the black crosses) along the region boundary via a geodesic, i.e., great circle. We then choose the middle of some great circle arc that is inside the polygon (inset, white star), testing whether an arc is inside using the orientation of the boundary curve tangents at the arc’s endpoints. Note that by construction the circle will necessarily either intersect the region boundary or be tangent to it. If it is tangent, we choose another arbitrary boundary point and repeat the procedure. 

%In the case that the point sampled is inside the largest discretized region, but outside the corresponding analytical region, due to boundary discretization error, $\chi$ is likely incorrect.  Given that there a few regions, we simply repeat the point sampling process in the largest region until we get a point far enough from the region's boundary.  This procedure is optional, and may be discarded at performance cost.

\textbf{Multiple Boundaries.} In the case that the surface has multiple disconnected boundaries, we project all the boundaries in a single spherical arrangement. 

\textbf{Exact Computations}. For parametric geometry, our pipeline enables exact computation of the winding number. As we show in \new{Appendix \ref{app:sos}}, the intersection problem (Eq.~\ref{eq:ray_bezier_int}), as well as the problems of intersecting two parametric curves, such as B\'ezier, and finding their self-intersections, can be relaxed into \emph{Sum-of-Squares} (SOS) formulations \cite{SOS}, which then yield convex problems solvable by standard semidefinite programming (SDP) solvers. This relaxation is tight, yielding exact recovery in almost all cases (Sec.~\ref{sec:validation}, \cite{SOS}). All these techniques, as shown in \cite{SOS} are easily extendable to arbitrary NURBS patches. In this case, for exactness, instead of sampling a point within the largest region on the sphere, we sample a random point on the sphere and perform a point location within the spherical region decomposition, which can be done exactly. SOS optimization, however, is slow (see Sec.~\ref{sec:validation}). Finally, for a region bounded by a parametric curve $\gamma(s): [0, L] \rightarrow S_p^2$ projected onto the unit sphere, using Gauss-Bonnet theorem, we can compute the area inside the curve as $A = 2\pi - \int_0^L k_g(s)$, where $k_g$ is the geodesic curvature of the boundary. Note that for most parametric boundaries, including B\'ezier curves, this integral has to be evaluated numerically. 

\section{Curve Networks}
\label{sec:gaps}

In this section, we demonstrate an application of our algorithm  determining winding number of a curve network \new{(Fig.~\ref{fig:vox}, \ref{fig:harmonic})}, where each patch satisfies a known linear PDE. In the following, we focus on the Laplace's equation for the simplicity of exposition; our method can be generalized to any PDE where Boundary Element Method is applicable, such as Poisson equation, biharmonic equation, or others.
\begin{figure}
	\includegraphics[width=\linewidth]{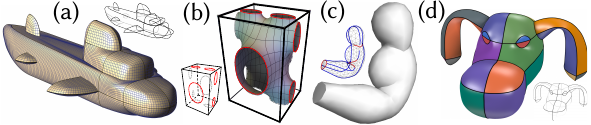}
	\caption{Curve networks can, depending on the context, represent different surfaces, such as (a) parametric patches \cite{DesigndrivenQuad}, (b) minimal surfaces \cite{wangComputingMinimalSurfaces2021}, \cite{MinimalSurfaceBlog} (c) other PDEs \cite{nealenFiberMeshDesigningFreeform2007}, or (d) more complex surfaces \cite{yuPiecewisesmoothSurfaceFitting2022}.}
	\label{fig:motivation}
\end{figure}

\subsection{Preprocessing}
\label{sec:preprocessing}
Given an input curve network composed of polylines, we use the method of \cite{zhuangGeneralEfficientMethod2013} to detect loops, consisting of input curve segments. We then identify \new{\emph{boundaries}} as \new{segments} adjacent to only one loop. If the surface is an orientable 2-manifold, we then consistently orient the loops via breadth-first search such that each non-boundary curve segment is traversed once in each direction. This procedure detects if the surface is not a manifold, in which case we compute a winding number as a sum of winding numbers of each manifold patch.
%Note that the method of \mycite{zhuangGeneralEfficientMethod2013} requires additional annotations for non-manifold curve networks; we assume such input is provided whenever necessary. 
%Our method on its own does not require any additional annotations, once the loops are identified. 

\subsection{Ray-Surface Intersection}
\label{sec:ray_int}
We need to find a point on the ray $\overrightarrow{pq}$ that satisfies the boundary value problem for a linear PDE for each coordinate $(x,y,z)$.  

We formulate the Dirichlet problem for Laplace's equation for $x: \Omega \subset \mathbb{R}^2 \rightarrow \mathbb{R}$:
\begin{align}
	\begin{split}
		&\Delta x = 0\\
	%\end{split}
	%\begin{split}
		&\left.x\right|_{\partial \Omega} = \tilde{x},
	\end{split}
	\label{eq:laplace}
\end{align}
where $\tilde{x}$ is the given curve loop $x$ coordinate. For simplicity, for the examples we set $\Omega = [0,1]^2$. \new{Solutions to this problem are related to minimal surfaces \cite{meeksSurveyClassicalMinimal2012}}.

%Alternatively, the same problem can be formulated for the biharmonic equation with both Dirichlet and Neumann boundary conditions:

%\begin{align}
%	&\Delta^2 x = 0\\
%	&\left.x\right|_{\partial \Omega} = \tilde{x},\\
%	&\left.\frac{\partial x}{\partial n}\right|_{\partial \Omega} = \tilde{w},
%	\label{eq:biharmonic}
%\end{align}

%where $w$ is a known normal derivative at the boundary (Neumann boundary condition).

Here we follow the Boundary Element Method (BEM) with collocation \cite{BEMcourse}. We express the surface as a smooth map $f: \mathbb{R}^2 \rightarrow \mathbb{R}^3$, where each coordinate of $f(\xi) = (x,y,z)$ is the solution to the corresponding boundary value problem (Eq.~\ref{eq:laplace}). %Note that even though the surface is in 3D, the parametric domain is in 2D, so all the Laplacian operators are regular 2D Euclidean Laplacians. 
In BEM, a representation formula states that the solution of the linear differential equation at a point $\xi \in \Omega$ is expressible as a boundary integral over $\Gamma = \partial \Omega$, which converts our surface into parametric form:
\begin{equation}
	f(\xi) = \int _{\eta \in \Gamma} G(\xi,\eta) \frac{\partial f(\eta)}{\partial n} - \frac{\partial G (\xi,\eta)}{\partial n} f(\eta) d\Gamma,
	\label{eq:representation}
\end{equation}
where $G(\xi,\eta)$ is the fundamental solution of the PDE, and $n$ is the outwards normal of $\Gamma$ (See Appendix \ref{app:bem_details}). 
\begin{figure}
	\centering
	\includegraphics[width=0.7\linewidth]{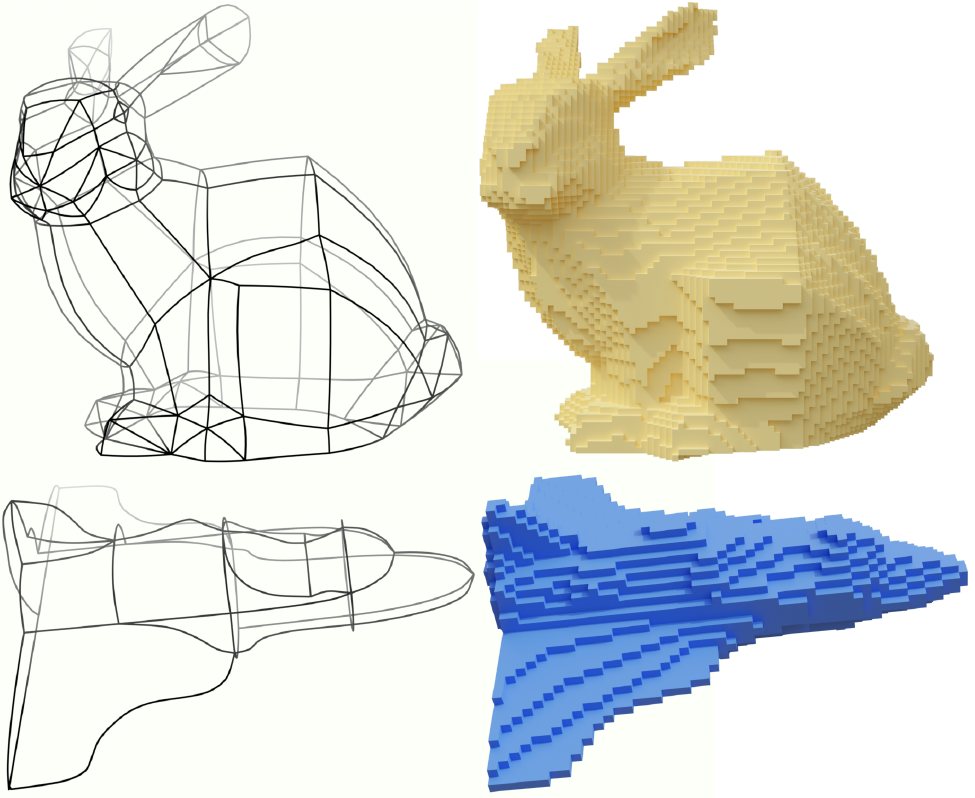}
	\caption{The results of our algorithm on closed curve networks, surfaces with minimal surfaces, with our one-shot algorithm. Voxels with a GWN of 1 are occupied and those with a GWN of 0 are left empty.}
	\label{fig:vox}
\end{figure}
%and for the 2D biharmonic equation
%\begin{equation}
%G(\xi,\eta) = \frac{| |\xi-\eta| |^2}{8\pi}\mathrm{ln}| |\xi-\eta| |.
%\label{eq:fsol_biharmonic}
%\end{equation}

%The representation formula effectively converts the surfaces expressed via a PDE into a parametric form.
%Using the representation formula, \misha{finish me}

%To compute the sign of an intersection, we look at the sign of the dot product between the normal of the surface at the intersection point, computed via differentiating \ref{eq:representation}, and the ray direction (Eq.~\ref{eq:chi}). %e.g.: 

%\begin{equation}
%	\frac{\partial f}{\partial u} = \int _{\eta \in \Gamma} \frac{\partial G(\xi,\eta)}{\partial u} \frac{\partial f(\eta)}{\partial n} - \frac{\partial^2 G (\xi,\eta)}{\partial n \partial u} f(\eta) d\Gamma,
%	\label{eq:drepresentation}
%\end{equation}

%where  $\frac{\partial G(\xi, \eta)}{\partial \xi}$ and $\frac{\partial^2 G (\xi,\eta)}{\partial n \partial \xi}$ are easily calculated in closed form 
%(see Eqs.~\ref{eq:fsol_laplace},\ref{eq:fsol_biharmonic}).
%(see Eq.~\ref{eq:fsol_laplace}).

We discretize the integrals in Eq.~\ref{eq:representation} using constant boundary elements, and numerically integrate for each polyline segment using trapezoid integration. To find Neumann boundary conditions, i.e., the normal derivative $\left.\frac{\partial f(\xi)}{\partial n} \right|_\Gamma$, we use the Dirichlet boundary conditions, which are the curve loop's coordinates, and  solve a dense linear system \cite{BEMcourse} via Cholesky decomposition. 
%For order four PDEs, such a biharmonic equation, those boundary conditions are independent, so for demonstration purposes we set them to zero. \misha{Is this even true? Is the representation formula still ok for biharmonic or other elliptic eqns?} 
%Our method does not depend on a particular choice of boundary conditions.

%\subsubsection{Finding all roots} 
%\label{sec:finding_all_roots}
Once both boundary conditions are known, we apply our method using the representation formula in Eq.~\ref{eq:representation} as a parametric form. We further optimize the process by observing that 
% or Biharmonic, Eqs.~\ref{eq:laplace},~\ref{eq:biharmonic}, 
a minimal surface is contained inside a bounding box of the boundary $\tilde x$ due to the maximum principle. 
%Formally, this follows from the maximum principle for the elliptic equations: as $\Omega$ is a compact set, the maximum and minimum of each coordinate $x,y,z$ has to be reached on the boundary $\partial \Omega$. 
We can therefore easily compute the minimum and maximum possible $t$ values $t_{\textrm{min}}, t_{\textrm{max}}$ by intersecting the ray with the axis-aligned bounding box, improving performance.

\section{Validation and Results}
\label{sec:results}
\label{sec:validation}

We implemented our method in C++ with Eigen \cite{eigenweb}, \verb|sphericalpolygon| library \cite{sphericalpolygon} for areas, and GNU Scientific Library's \cite{gsl_manual} implementation of Powell’s Hybrid method as a nonlinear solver for intersection with parametric geometry. 

Throughout the paper, we demonstrate the results of our one-shot winding algorithm on various 2D parametric curves, parametric surfaces (a set of Coons patches in Fig.~\ref{fig:teaser} a and \ref{fig:parametric}b, extrusion surfaces in Fig.~\ref{fig:parametric_gaps}d,\ref{fig:parametric}a), meshes (Fig.~\ref{fig:teaser}c, \ref{fig:mesh}), and curve networks with minimal surfaces (Fig.~\ref{fig:teaser}b, \ref{fig:harmonic}). We demonstrate surfaces with boundary, because without boundary our method is a trivial inside-outside test. The boundaries on the curve networks are selected manually. In Fig.~\ref{fig:vox}, we show a typical application of winding numbers, a voxelization of the volume enclosed inside the curve network, computed without discretizing the surfaces. We consider a voxel inside if its center has a winding number $\geq 0.5$. We demonstrate boolean operations in Fig.~\ref{fig:bool}.  To compute the `set union', we add the winding numbers; to compute `set intersection', we multiply them pointwise.

\begin{figure}
	\includegraphics[width=\linewidth]{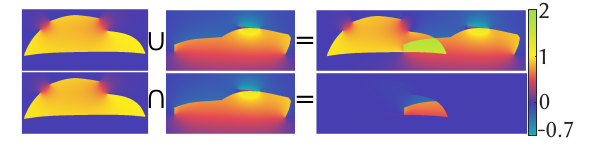}
	\caption{Winding numbers enable boolean operations on curve networks. Here `set union' is sum of winding numbers and `intersection' is the pointwise product.}
	\label{fig:bool}
\end{figure}

%We additionally validate our algorithmic choices, for both parametric representations and meshes, in various ways. 

\textbf{2D Parametric Geometry.} We first evaluate the performance and precision of our method on 2D parametric geometry. As a test, we randomly generate cubic B\'ezier curves (N=1000) with control points in the unit square. For each curve, we compute a winding number for each point of a regular grid (\new{$250\times250$}) in the curve's bounding box (Fig.~\ref{fig:performance_2d}). 

\begin{figure*}
	\includegraphics[width=\linewidth]{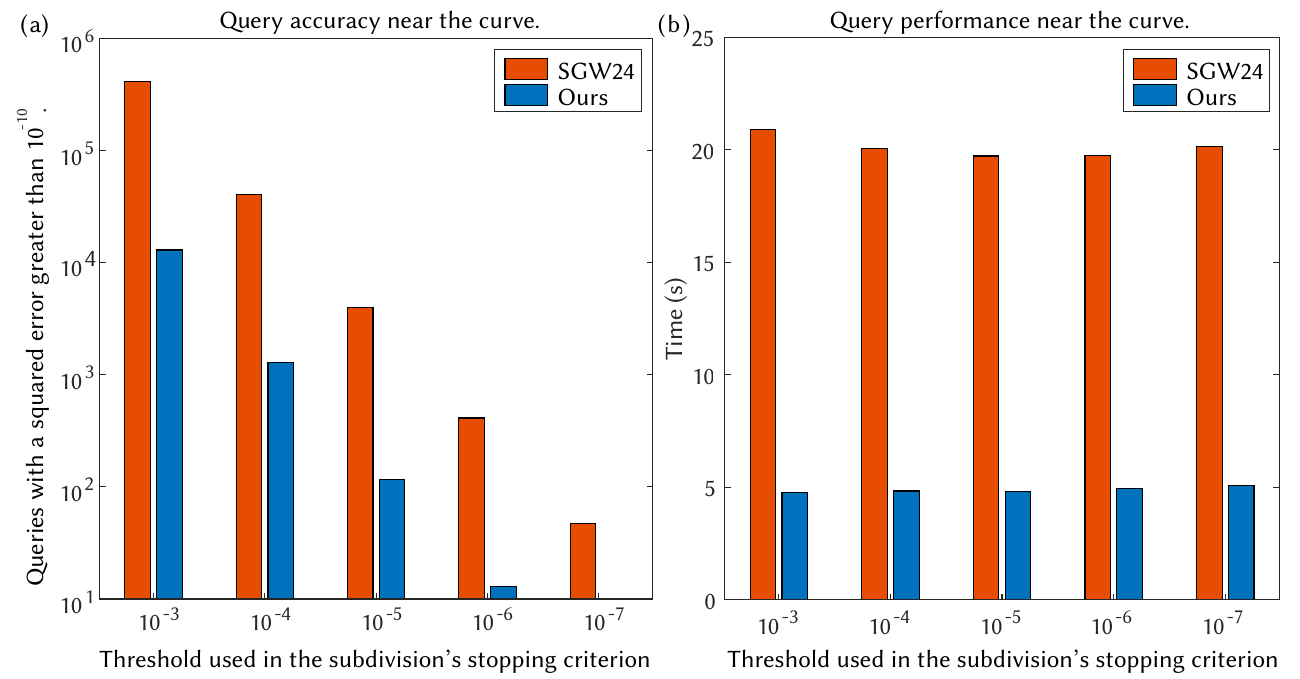}
	\caption{For 2D parametric curves, on a typical regular grid of query points, our method (a) is more precise and faster than the state-of-the-art method \cite{RobustContainmentQueries2024a} (b). The performance of both methods \new{is minimally affected by} the stopping criterion threshold, since most of the computation time is spent calculating arccosines. We find that our method is \new{$4\times$} faster on a regular grid of \new{$250\times250$} query points within the bounding box of 1000 random cubic B\'{e}zier curves. Both methods fall back to the same algorithm when the query point is outside the bounding box of the curve.}
	\label{fig:performance_2d}
\end{figure*}

We compare our performance with the 2D state-of-the-art method \cite{RobustContainmentQueries2024a}. For any query point outside the bounding box, we use the same closure property as \cite{RobustContainmentQueries2024a}, so our performance on those points is identical. Inside the bounding box, on a regular grid of query points, however, our method performs significantly faster (up to \new{$4\times$}), as we only need to perform subdivisions for one ray per row due to our one-shot strategy. Furthermore, in this setup, our algorithm is more precise: We obtain 0 `misclassified' points (winding number error above $10^{-10}$ compared to the ground truth) with a threshold of $10^{-7}$. Our method takes 5074ms,while \cite{RobustContainmentQueries2024a} takes 20131ms and produces 47 misclassified points.

For \emph{random} query points, without a regular grid, we cannot reuse the rays and have to shoot a new ray for each query point. In this case, our method is roughly 2x slower for the same accuracy compared to \cite{RobustContainmentQueries2024a}. On the same test set, per uniformly sampled query point in the bounding box, our method takes 0.08$\mu s$, while their method takes 0.038$\mu s$. We underline, however, that regularly sampled grids of query points are quite common in some applications, including voxelization in 3D or visualization in 2D/3D.

\textbf{Parametric Surfaces.} We validate the accuracy of our one-shot winding number computation for a discretized boundary and compare it with the accuracy of the hierarchical winding number algorithm (HWN) \cite{Jacobson-13-winding} and \cite{barillFastWindingNumbers2018} that use a mesh (Fig.~\ref{fig:accuracy_err}). As as simple test case, we took a randomly generated B\'ezier triangle, generated $10^4$ query points evenly in an arbitrary slice, and plotted maximum and mean $\ell_2$ (RMSE) errors \emph{in a logarithmic scale} of our method and theirs as functions of the number of boundary samples and mesh vertices, respectively. For the mesh-based methods, we meshed the parametric domain with a regular triangle mesh. The horizontal axis indicates the number of boundary edges for our method and their mesh (bottom), and the corresponding number of faces (top).

Even with a small number of boundary edges, our method is quite precise. Both mesh-based methods, until a certain refinement level, have large maximum errors close to the surface: Each time a query point is on one side of the parametric surface and on the other of the mesh, the winding number will have a wrong sign, yielding a significant error. In contrast, the error of our method, even for roughly discretized boundaries, is significantly smaller, noticeable only around the surface boundary, and stems mainly from approximation of the spherical areas. %Using SOS completely eliminates all others sources of error.

Under refinement, FWN \cite{barillFastWindingNumbers2018}, being an approximation, does not improve. Our method becomes an order of magnitude more precise than the exact mesh-based \cite{Jacobson-13-winding}; our worst (maximum) error is roughly the same as their average. Note that to get to a similar maximum error accuracy as our $1900$ boundary edges, mesh-based methods would need roughly $16.3\cdot 10^6$ faces, taking 0.9Gb in memory. Our memory footprint is negligible in comparison (kilobytes).

\begin{figure}
	\includegraphics[width=\linewidth]{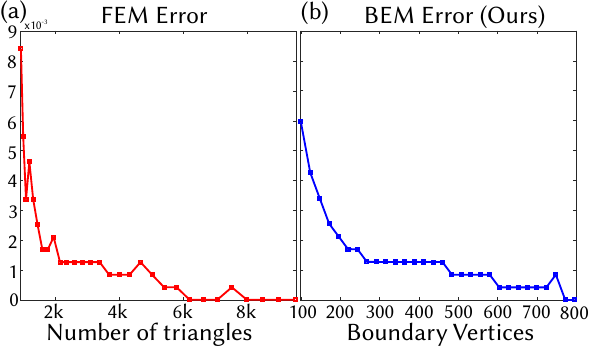}
	\caption{Accuracy of the the winding number computation for curve networks with minimal surfaces computed via BEM. (right). Our method, in general, has similar precision to the mesh-based method (left).}
	\label{fig:fig_error}
\end{figure}

\begin{figure}
	\includegraphics[width=\linewidth]{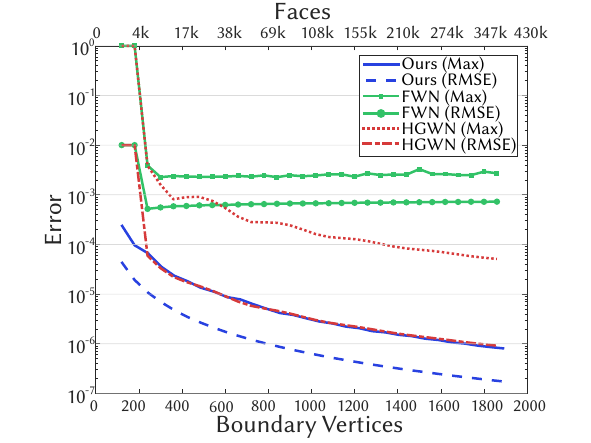}
	\caption{Logarithmic error plot with respect to ground truth for our method (blue), \cite{Jacobson-13-winding} (HGWN, red) and \cite{barillFastWindingNumbers2018} (FWN, green). Even when the boundary is discretized, our method computes winding numbers a few orders of magnitude more precisely than \cite{Jacobson-13-winding} or \cite{barillFastWindingNumbers2018} on a mesh of comparable complexity. Horizontal axis shows the number of faces (top) and boundary vertices (bottom) of a mesh and used in our method.}
	\label{fig:accuracy_err}
\end{figure}

\textbf{3D Performance.} We evaluate performance of our method on a Intel® Core™ i7-9700X @ 4.900GHz. We measure all the total performances using the C++ code, parallelized via OpenMP.

We present the results in Table~\ref{tab:timing_bem_parametric} for parametric and minimal surfaces, computing winding numbers of all points in a slice (i.e., the optimization in Sec.~\ref{sec:impl_details} is on). The performance per point depends mostly on the number of boundary samples; for most of our results a query point takes $\approx 0.2$ms. Dense system solve for BEM is done once for all the query points and takes 10-90ms.

The bottleneck of our algorithm is marked as `Boundary Processing' (BP); almost all time is spent on computing finding all intersections between spherical segments. Note that despite the optimal algorithmic complexity of $O(B \log B)$, in our implementation we use a na\"ive $O(B^2)$ algorithm which may explain the performance.

For meshes, the accuracy of our method is exactly the same as the ground truth of \cite{Jacobson-13-winding}. Performance-wise, our method may bring advantages for meshes where, as discussed in Sec.~\ref{sec:impl_details}, the boundary has few elements, while the mesh has many faces. In Table~\ref{tab:timing_mesh} we present typical performance statistics for such meshes. Note that the approximation algorithm \cite{barillFastWindingNumbers2018} timing statistics are a few times faster: $0.59, 0.55, 0.53\mu$s. Our goal, however, is precise winding number computation: as Fig.~\ref{fig:accuracy_err} shows, their approximation is often imprecise.
 
\begin{table}[t]
	\centering
	\begin{tabular}{|c|c|c|c|c|c|}
		\hline
		Input & \# L & \# B & \# Q & \% B.P. & Ours (ms) \\
		\toprule 
		\hline
		Beetle & 53 & 950 & 196k/308k  & 93.4 & 0.157\\
		\hline
		Hand & 132 & 2150 & 136k/136k & 91.92 & 0.26\\
		\hline
		Enterprise & 117 & 1632 & 250k & 99.96 & 0.26\\
		\hline
		Spacecraft & 20 & 2761 & 113k & 98.8 & 1.20\\
		\hline
		\hline
		Roll & 1 & 200 &  250k & 95.93 & 0.10\\
		\hline
		Car & 25 & 1500 & 180k & 28.92 & 1.13\\
		\hline
		Ship  & 30 & 1780 & 262k & 20.98 & 1.84\\
		\hline
		Gear & 2 & 400 & 250k & 41.66 & 0.208\\
		\hline
		\hline

	\end{tabular}
	\begin{tabular}{|c|c|c|c|c|c|c|}
		\hline
		Input & \# F & \# B & \# Q & \% B. P. & HG ($\mu$s) & Ours ($\mu$s)\\
		\toprule 
		\hline
		Piggy & 11k & 6 & 600k & 96.66 & 4.41 & 1.46 \\
		\hline
		Bunny & 5k & 9 & 360k & 97.64  & 2.69 &  2.12 \\
		\hline
		Camel & 23k & 56 & 480k & 99.01 &  4.71 & 8.99 \\
		\hline
		\hline
	\end{tabular}

	\caption{Timing statistics for curve networks with minimal surfaces (top half) and parametric surfaces with Coons patches (Car, Ship), and extrusion surfaces (gear, roll). The columns $\#L, \#B, \#Q$ refer to loops, boundary edges, and query points respectively.  Total time in ms per query point is in the last column; time spent boundary processing is \%BP. Bottom: statistics for meshes, time in $\mu$s ($10^{-6} s$) per query point. HG refers to the hierarchical method of \cite{Jacobson-13-winding}.} 
	\label{tab:timing_bem_parametric} 
		\label{tab:timing_mesh}
\end{table}

\textbf{Curve networks with minimal surfaces.} We compare the accuracy of the nonlinear solver approach for a BEM curve network, with the method of \cite{Jacobson-13-winding} on a FEM discretization of the surface using linear elements. As the ground truth surface, we take $z(x,y) = e^{x-1} \sin(y) - e^x  \cos(y)$; we mesh it and compute its winding numbers. We then reconstruct the surface by FEM using a regular grid on the parametric domain. Then we uniformly sample $342$k points, keeping only those 2368 points with the absolute value of the ground truth winding number $\in (0.4995, 0.5005)$; on the rest the error is 0. We can see (Fig.~\ref{fig:fig_error})  that our solution is as accurate as theirs for just a fraction of the stored data (e.g., 500 vertices vs. 5000 triangles). The discrete increments of the error are because each query point can be classified as either on the correct side of the surface, yielding zero error, and on the wrong side, yielding a winding number error of 1.

\textbf{Adaptive Quadrature.} 
\label{sec:adaptive_quadrature}
For parametric curves and surfaces, the generalized winding number can be computed via numerically evaluating \new{the} signed solid angle surface integral using adaptive quadrature. In 2D, for a regular grid of query points, we are significantly faster than the state-of-the-art method \cite{RobustContainmentQueries2024a}, which was already an improvement over adaptive quadrature. In 3D, our approach is generally slower than adaptive quadrature methods, with the bottleneck being the computation of spherical arrangement.

For a class of parametric surfaces, however, with a simpler boundary, for instance, planar, our method is faster in a scenario with a regular grid. As a projection of a planar curve on a sphere can have no self-intersections, for such parametric surfaces we only need to perform a trivial area computation and a ray-surface intersection. We show two examples of surfaces of revolution with planar boundaries in Fig.~\ref{fig:surf_revolution} where we are both faster and more accurate than adaptive quadrature. On a regular grid of $800\times800$, it is 83s for our method vs. 103s for adaptive quadrature (both yielded similar errors of 3.1e-5, Fig.~\ref{fig:surf_revolution}a). Since the parametric surface in Fig.~\ref{fig:surf_revolution}b has a sharp feature, we switched both methods to higher precision arithmetic, decreasing performance. On the challenging example in Fig.~\ref{fig:surf_revolution}b, our method computed the winding numbers in 484s with an MSE of 2.6e-4. Adaptive quadrature took 12637s to compute result with a high MSE of 0.0809. This showcases the critical difference with the adaptive quadrature: since our method does not need to discretize the surface, high-frequency surface details are much less of a problem for us. Our approach also requires significantly less function evaluations compared to adaptive quadrature.

\begin{figure}
	\includegraphics[width=\linewidth]{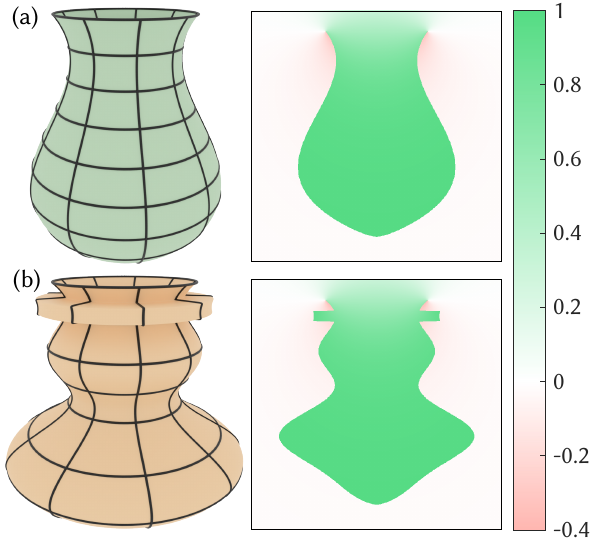}
	\caption{Our approach on two surfaces of revolutions. Computing the winding numbers for the surface in (a) took our method 83s while adaptive quadrature took 103s for a similar error of  $3.125e-5$ and $3.140e-5$ respectively.
		In (b), as the surface has a sharp feature, we used higher precision arithmetic, significantly decreasing performance. 
		For this experiment, our method took 484s while having a MSE of $2.6875e-4$. This was an almost intractable input for adaptive quadrature, taking it 12637s to only achieve a MSE of 0.0809.}
	\label{fig:surf_revolution}
\end{figure}

\textbf{Singular and unbounded surfaces.} Our method is compatible with both singular and unbounded surfaces, as long as the boundary has an asymptote at infinity and each ray has a finite number of intersections. We demonstrate an example in Fig.~\ref{fig:singularity_unbounded}. For unbounded surfaces, we project the asymptote of the surface at $\lim_{r \rightarrow \infty}$ from the center of the parametric domain, otherwise our method requires no changes. Note that these surfaces are hard to process with conventional methods, as discretizing them may result in significant accuracy loss.

\begin{figure}
	\includegraphics[width=\linewidth]{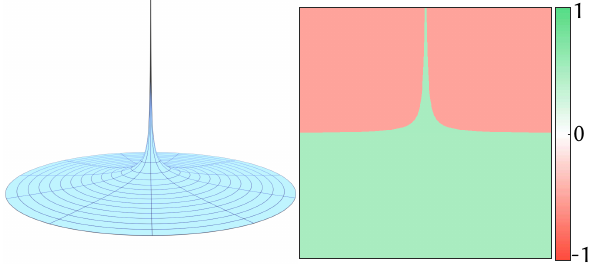}
	\caption{Our method works without modification on any oriented surfaces that support ray-intersection queries, including infinite surfaces with singularities. }
	\label{fig:singularity_unbounded}
\end{figure}

\textbf{Precise computation for parametric surfaces}. We also implemented \new{a proof of concept} SOS optimization for precise winding number computation for parametric shapes, in MATLAB. We use YALMIP \cite{Lofberg2004} as the SOS interface and MOSEK \cite{mosek} to solve.

We validate the accuracy of the SOS solvers, showing all the intersections with parametric surfaces can be found within machine precision. Note that B\'ezier curves and B\'ezier triangles were chosen as an example; \cite{SOS,SOS2} demonstrate that all these techniques are easily generalizable to arbitrary NURBS patches.

In each test, we generate the random control points, as well as ray origin and direction, following a uniform $[-1,1]$ distribution. We generated 400 test pairs (half with intersections, half without) for each test. For the ray intersections, we got exact recovery (within machine precision) or correctly identified no intersections for $99.75\%$ of the rays. Exact recovery is identified as a corresponding eigenvalue being larger than a threshold ($10^{-3}$ in our experiments) \cite{SOS}. The only ray that did not have exact recovery had an error in $t$ of 0.0064; such rare ray can be simply ignored using the eigenvalue threshold. The projected B\'ezier curve intersection test has a similar exact recovery rate of 99.75\%, same for self-intersections.

SOS, while being precise, is slow: finding intersections of projected B\'ezier curves takes $592 \pm 275 ms$ via SOS versus $1.2$ms via the polyline intersection code; ray-B\'ezier triangle intersection takes $581\pm 25$ms via SOS versus $8.85$ms via the nonlinear solver with our parameters. Most of SOS time is spent in YALMIP, not in the solver. 

\textbf{Limitations.} As the complexity of our method is driven by the complexity of the boundary, our method excels when the boundary is simple, and is slow when the boundary is complex or has multiple connected components. Perhaps an efficient system would choose a suitable method to compute the winding number based on the boundary complexity vs. the complexity of the curve/surface itself, falling back to previous methods if necessary.  Additionally, errors in root finding, especially missed or spurious roots, can
change the value of $\chi$, causing errors. While our analysis (Sec.~\ref{sec:validation}) shows this is extremely rare, this can be fixed by increasing the search resolution or, when exactness is necessary, via SOS.

%When not using SOS, the accuracy of our computations relies on the nonlinear equation solver that we use. Errors in root finding, especially missed or spurious roots, can
%change the value of $\chi$, causing errors. While our analysis (Sec.~\ref{sec:validation}) shows that this rare, for precision applications one should use SOS.

\section{Conclusions}

We have presented a new method for computing a generalized winding number for parametric curves and surfaces, meshes, and curve networks with minimal surfaces. We show that for many cases such as 2D parametric curves or 3D meshes or parametric surfaces with simple boundary, on a regular grid, our method is faster than the state of the art. With the use of SOS, our method can compute winding numbers of parametric geometry precisely. We hope that the future work will explore applications of our method to tracking winding numbers of moving or optimized geometry, especially if the boundary remains fixed, where discretization becomes even more of an issue.

%\new{We hope that the future work will explore applications of our method to tracking winding numbers of moving geometry, where discretization becomes even more of an issue, whereas our one-shot strategy should require only tracking the ray-surface intersections and updating the boundary arrangement on a sphere; both can be done efficiently and precisely.}

\begin{figure*}
	\includegraphics[width=\linewidth]{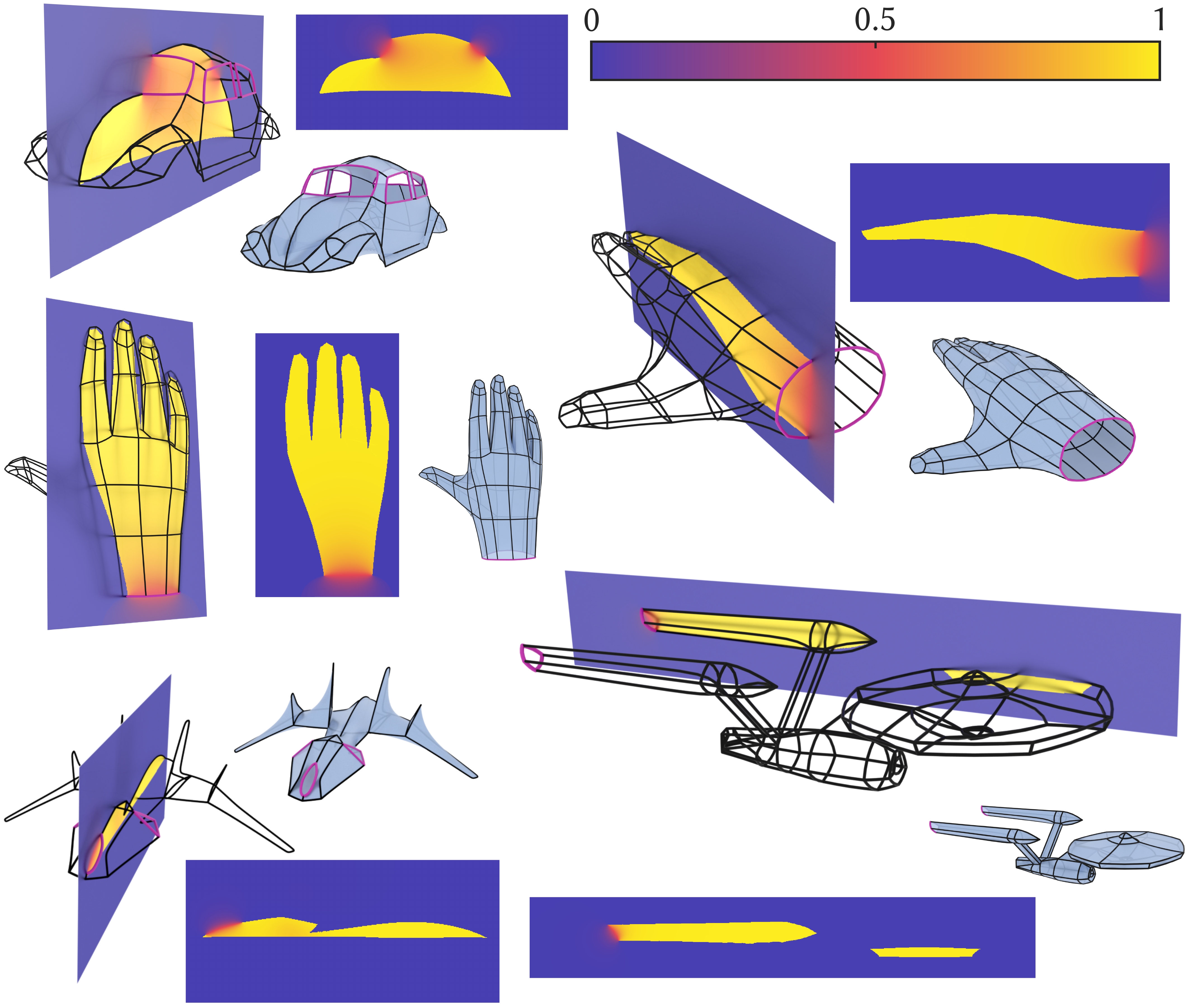}
	\caption{Winding numbers computed for each point of a planar slice for the curve networks, surfaced with minimal surfaces, with our one-shot algorithm. Boundaries are highlighted in pink.}
	\label{fig:harmonic}
\end{figure*}

\section*{Acknowledgements}
We acknowledge the support of the Natural Sciences and Engineering Research Council of Canada (NSERC) under Grant No.: RGPIN-2019-05097 (“Creating Virtual Shapes via Intuitive Input”) and RGPIN-2024-04968 ("Modelling and animation via intuitive input"), and the NSERC - Fonds de recherche du Québec - Nature et technologies (FRQNT) NOVA Grant No. 314090. We thank Ivan Puhachov for providing several figures and Paul Zhang for his valuable assistance with Sum-of-Squares.

\bibliographystyle{eg-alpha-doi} 
\bibliography{wn}  

\newcommand{\etalchar}[1]{$^{#1}$}
\begin{thebibliography}{\uppercase{XDW{\etalchar{*}}23}}

\bibitem[ApS24]{mosek}
\textsc{ApS M.}:
\newblock \emph{The MOSEK optimization toolbox for MATLAB manual. Version
  10.1.}, 2024.
\newblock URL: \url{http://docs.mosek.com/latest/toolbox/index.html}.

\bibitem[BA18]{binyshMaxwellTheorySolid2018}
\textsc{Binysh J., Alexander G.~P.}:
\newblock Maxwell's theory of solid angle and the construction of knotted
  fields.
\newblock \emph{Journal of Physics A: Mathematical and Theoretical 51}, 38
  (Sept. 2018), 385202.
\newblock \href {https://doi.org/10.1088/1751-8121/aad8c6}
  {\path{doi:10.1088/1751-8121/aad8c6}}.

\bibitem[BCKO08]{ComputationGeometryAlgorithms}
\textsc{Berg M.~d., Cheong O., Kreveld M.~v., Overmars M.}:
\newblock \emph{Computational Geometry: Algorithms and Applications}, 3rd
  ed.~ed.
\newblock Springer-Verlag TELOS, Santa Clara, CA, USA, 2008.

\bibitem[BDS{\etalchar{*}}18]{barillFastWindingNumbers2018}
\textsc{Barill G., Dickson N.~G., Schmidt R., Levin D. I.~W., Jacobson A.}:
\newblock Fast winding numbers for soups and clouds.
\newblock \emph{ACM Transactions on Graphics 37}, 4 (Aug. 2018), 1--12.
\newblock \href {https://doi.org/10.1145/3197517.3201337}
  {\path{doi:10.1145/3197517.3201337}}.

\bibitem[BFP{\etalchar{*}}11]{becciu3DWindingNumber2011}
\textsc{Becciu A., Fuster A., Pottek M., Van Den~Heuvel B., Ter Haar~Romeny B.,
  Van~Assen H.}:
\newblock {{3D Winding Number}}: {{Theory}} and {{Application}} to {{Medical
  Imaging}}.
\newblock \emph{International Journal of Biomedical Imaging 2011} (2011),
  1--13.
\newblock \href {https://doi.org/10.1155/2011/516942}
  {\path{doi:10.1155/2011/516942}}.

\bibitem[BWSS12]{DesigndrivenQuad}
\textsc{Bessmeltsev M., Wang C., Sheffer A., Singh K.}:
\newblock Design-driven quadrangulation of closed {{3D}} curves.
\newblock In \emph{{{ACM Transactions}} on {{Graphics}}} (2012), vol.~31.
\newblock \href {https://doi.org/10.1145/2366145.2366197}
  {\path{doi:10.1145/2366145.2366197}}.

\bibitem[CI24]{chern2024area}
\textsc{Chern A., Ishida S.}:
\newblock Area formula for spherical polygons via prequantization.
\newblock \emph{SIAM Journal on Applied Algebra and Geometry 8}, 3 (2024),
  782--796.
\newblock URL: \url{https://doi.org/10.1137/23M1565255}, \href
  {http://arxiv.org/abs/https://doi.org/10.1137/23M1565255}
  {\path{arXiv:https://doi.org/10.1137/23M1565255}}, \href
  {https://doi.org/10.1137/23M1565255} {\path{doi:10.1137/23M1565255}}.

\bibitem[FGC23a]{PerspectivesWindingNumbers2023}
\textsc{Feng N., Gillespie M., Crane K.}:
\newblock Perspectives on {{Winding Numbers}}, 2023.
\newblock URL:
  \url{https://markjgillespie.com/Research/WNoDS/PerspectivesOnWindingNumbers.pdf}.

\bibitem[FGC23b]{fengWindingNumbersDiscrete2023}
\textsc{Feng N., Gillespie M., Crane K.}:
\newblock Winding {{Numbers}} on {{Discrete Surfaces}}.
\newblock \emph{ACM Transactions on Graphics 42}, 4 (Aug. 2023), 1--17.
\newblock \href {https://doi.org/10.1145/3592401} {\path{doi:10.1145/3592401}}.

\bibitem[GDT{\etalchar{*}}96]{gsl_manual}
\textsc{Galassi M., Davies J., Theiler J., Gough B., Jungman G., Alken P.,
  Booth M., Rossi F., Johnson M., Rossi G., Moore S., et~al.}:
\newblock \emph{GNU Scientific Library Reference Manual}, 3rd~ed.
\newblock ISBN 0954612078, 1996.
\newblock URL: \url{http://www.gnu.org/software/gsl/}.

\bibitem[GHL{\etalchar{*}}20]{gryaditskaya2020lifting}
\textsc{Gryaditskaya Y., H{\"a}hnlein F., Liu C., Sheffer A., Bousseau A.}:
\newblock Lifting {{Freehand Concept Sketches}} into {{3D}}.
\newblock \emph{ACM Transactions on Graphics (Proceedings of Siggraph Asia)}
  (2020).

\bibitem[GJ{\etalchar{*}}10]{eigenweb}
\textsc{Guennebaud G., Jacob B., et~al.}:
\newblock Eigen v3.
\newblock http://eigen.tuxfamily.org, 2010.

\bibitem[{Goo}23]{GoogleTiltBrush2023}
\textsc{{Google}}:
\newblock Google {{TiltBrush}}.
\newblock Google, 2023.

\bibitem[GP10]{guilleminDifferentialTopology2010}
\textsc{Guillemin V., Pollack A.}:
\newblock \emph{Differential {{Topology}}}, reprint edition~ed.
\newblock {American Mathematical Society}, {Providence, R.I}, Aug. 2010.

\bibitem[HBW03]{houstonUnifiedApproachModeling}
\textsc{Houston B., Bond C., Wiebe M.}:
\newblock A unified approach for modeling complex occlusions in fluid
  simulations.
\newblock In \emph{ACM SIGGRAPH 2003 Sketches \& Applications} (New York, NY,
  USA, 2003), SIGGRAPH '03, Association for Computing Machinery, p.~1.
\newblock URL: \url{https://doi.org/10.1145/965400.965561}, \href
  {https://doi.org/10.1145/965400.965561} {\path{doi:10.1145/965400.965561}}.

\bibitem[HZG{\etalchar{*}}18]{tetMeshingWild}
\textsc{Hu Y., Zhou Q., Gao X., Jacobson A., Zorin D., Panozzo D.}:
\newblock Tetrahedral meshing in the wild.
\newblock \emph{ACM Transactions on Graphics 37}, 4 (Aug. 2018), 1--14.
\newblock \href {https://doi.org/10.1145/3197517.3201353}
  {\path{doi:10.1145/3197517.3201353}}.

\bibitem[JKS13]{Jacobson-13-winding}
\textsc{Jacobson A., Kavan L., Sorkine O.}:
\newblock Robust {{Inside-Outside Segmentation}} using {{Generalized Winding
  Numbers}}.
\newblock \emph{ACM Trans. Graph. 32}, 4 (2013).

\bibitem[Kno18]{solid_angle_curve}
\textsc{Knoppel F.}:
\newblock {Tutorial 5 – Solid angle of space curves}.
\newblock
  \url{http://wordpress.discretization.de/ddg2018/2018/06/19/tutorial-5-solid-angle-of-space-curves/},
  2018.
\newblock [Online; accessed 26-04-2024].

\bibitem[LaF06]{BEMcourse}
\textsc{LaForce T.}:
\newblock {{PE281}} boundary element method course notes. {{Stanford}},
  {{CA}}., June 2006.

\bibitem[Li20]{sphericalpolygon}
\textsc{Li C.}:
\newblock Sphericalpolygon, 2020.
\newblock URL: \url{https://github.com/lcx366/SphericalPolygon}.

\bibitem[L{\"{o}}f04]{Lofberg2004}
\textsc{L{\"{o}}fberg J.}:
\newblock Yalmip : A toolbox for modeling and optimization in matlab.
\newblock In \emph{In Proceedings of the CACSD Conference} (Taipei, Taiwan,
  2004).

\bibitem[L{\"{o}}f09]{Lofberg2009}
\textsc{L{\"{o}}fberg J.}:
\newblock Pre- and post-processing sum-of-squares programs in practice.
\newblock \emph{IEEE Transactions on Automatic Control 54}, 5 (2009),
  1007--1011.

\bibitem[Mei69]{meisterGeneraliaGenesiFigurarum1769}
\textsc{Meister A. L.~F.}:
\newblock \emph{{Generalia de genesi figurarum planarum et inde pendentibus
  earum affectionibus}}.
\newblock {Novi Comm. Soc. Reg. Scient. Gotting.}, 1769.

\bibitem[MP12]{meeksSurveyClassicalMinimal2012}
\textsc{Meeks W., P{\'e}rez J.}:
\newblock \emph{A {{Survey}} on {{Classical Minimal Surface Theory}}}, vol.~60
  of \emph{University {{Lecture Series}}}.
\newblock American Mathematical Society, Providence, Rhode Island, Dec. 2012.
\newblock \href {https://doi.org/10.1090/ulect/060}
  {\path{doi:10.1090/ulect/060}}.

\bibitem[MST{\etalchar{*}}21]{NeRF21}
\textsc{Mildenhall B., Srinivasan P.~P., Tancik M., Barron J.~T., Ramamoorthi
  R., Ng R.}:
\newblock Nerf: representing scenes as neural radiance fields for view
  synthesis.
\newblock \emph{Commun. ACM 65}, 1 (Dec. 2021), 99–106.
\newblock URL: \url{https://doi.org/10.1145/3503250}, \href
  {https://doi.org/10.1145/3503250} {\path{doi:10.1145/3503250}}.

\bibitem[MZPS21]{SOS}
\textsc{Marschner Z., Zhang P., Palmer D., Solomon J.}:
\newblock Sum-of-squares geometry processing.
\newblock \emph{ACM Trans. Graph. 40}, 6 (dec 2021).
\newblock URL: \url{https://doi.org/10.1145/3478513.3480551}, \href
  {https://doi.org/10.1145/3478513.3480551}
  {\path{doi:10.1145/3478513.3480551}}.

\bibitem[Nee97]{needhamVisualComplexAnalysis1997}
\textsc{Needham T.}:
\newblock \emph{Visual {{Complex Analysis}}}.
\newblock {Clarendon Press}, 1997.

\bibitem[NISA07]{nealenFiberMeshDesigningFreeform2007}
\textsc{Nealen A., Igarashi T., Sorkine O., Alexa M.}:
\newblock {{FiberMesh}}: Designing freeform surfaces with {{3D}} curves.
\newblock \emph{ACM Transactions on Graphics 26}, 3 (2007), 1--8.
\newblock \href {https://doi.org/10.1145/1276377.1276429}
  {\path{doi:10.1145/1276377.1276429}}.

\bibitem[NT03]{nooruddinSimplificationRepairPolygonal2003}
\textsc{Nooruddin F., Turk G.}:
\newblock Simplification and repair of polygonal models using volumetric
  techniques.
\newblock \emph{IEEE Transactions on Visualization and Computer Graphics 9}, 2
  (Apr. 2003), 191--205.
\newblock \href {https://doi.org/10.1109/TVCG.2003.1196006}
  {\path{doi:10.1109/TVCG.2003.1196006}}.

\bibitem[SGW24]{RobustContainmentQueries2024a}
\textsc{Spainhour J., Gunderman D., Weiss K.}:
\newblock Robust {{Containment Queries}} over {{Collections}} of {{Rational
  Parametric Curves}} via {{Generalized Winding Numbers}}.
\newblock \emph{ACM Transactions on Graphics 43}, 4 (July 2024), 1--14.
\newblock \href {https://doi.org/10.1145/3658228} {\path{doi:10.1145/3658228}}.

\bibitem[Sor05]{sorkineLaplacianMeshProcessing}
\textsc{Sorkine O.}:
\newblock {Laplacian Mesh Processing}.
\newblock In \emph{Eurographics 2005 - State of the Art Reports} (2005),
  Chrysanthou Y., Magnor M., (Eds.), The Eurographics Association.
\newblock \href {https://doi.org/10.2312/egst.20051044}
  {\path{doi:10.2312/egst.20051044}}.

\bibitem[WC21]{wangComputingMinimalSurfaces2021}
\textsc{Wang S., Chern A.}:
\newblock Computing minimal surfaces with differential forms.
\newblock \emph{ACM Transactions on Graphics 40}, 4 (Aug. 2021), 1--14.
\newblock \href {https://doi.org/10.1145/3450626.3459781}
  {\path{doi:10.1145/3450626.3459781}}.

\bibitem[Web19]{MinimalSurfaceBlog}
\textsc{Weber M.}:
\newblock Out of the box.
\newblock \url{https://minimalsurfaces.blog/2019/01/13/out-of-the-box/}, 2019.
\newblock Accessed: 2024-01-22.

\bibitem[XBC19]{xiaoNoniterativeMethodRobustly2019}
\textsc{Xiao X., Bus{\'e} L., Cirak F.}:
\newblock A noniterative method for robustly computing the intersections
  between a line and a curve or surface.
\newblock \emph{International Journal for Numerical Methods in Engineering
  120}, 3 (2019), 382--390.
\newblock \href {https://doi.org/10.1002/nme.6136}
  {\path{doi:10.1002/nme.6136}}.

\bibitem[XDW{\etalchar{*}}23]{xuGloballyConsistentNormal2023}
\textsc{Xu R., Dou Z., Wang N., Xin S., Chen S., Jiang M., Guo X., Wang W., Tu
  C.}:
\newblock Globally consistent normal orientation for point clouds by
  regularizing the winding-number field.
\newblock \emph{ACM Trans. Graph. 42}, 4 (July 2023).
\newblock URL: \url{https://doi.org/10.1145/3592129}, \href
  {https://doi.org/10.1145/3592129} {\path{doi:10.1145/3592129}}.

\bibitem[YAB{\etalchar{*}}22]{yuPiecewisesmoothSurfaceFitting2022}
\textsc{Yu E., Arora R., B{\ae}rentzen J.~A., Singh K., Bousseau A.}:
\newblock Piecewise-smooth surface fitting onto unstructured {{3D}} sketches.
\newblock \emph{ACM Transactions on Graphics 41}, 4 (July 2022), 1--16.
\newblock \href {https://doi.org/10.1145/3528223.3530100}
  {\path{doi:10.1145/3528223.3530100}}.

\bibitem[ZGZJ16]{zhouMeshArrangementsSolid2016}
\textsc{Zhou Q., Grinspun E., Zorin D., Jacobson A.}:
\newblock Mesh arrangements for solid geometry.
\newblock \emph{ACM Transactions on Graphics 35}, 4 (July 2016), 1--15.
\newblock \href {https://doi.org/10.1145/2897824.2925901}
  {\path{doi:10.1145/2897824.2925901}}.

\bibitem[ZMST23]{SOS2}
\textsc{Zhang P., Marschner Z., Solomon J., Tamstorf R.}:
\newblock Sum-of-squares collision detection for curved shapes and paths.
\newblock In \emph{ACM SIGGRAPH 2023 Conference Proceedings} (New York, NY,
  USA, 2023), SIGGRAPH '23, Association for Computing Machinery.
\newblock URL: \url{https://doi.org/10.1145/3588432.3591507}, \href
  {https://doi.org/10.1145/3588432.3591507}
  {\path{doi:10.1145/3588432.3591507}}.

\bibitem[ZZCJ13]{zhuangGeneralEfficientMethod2013}
\textsc{Zhuang Y., Zou M., Carr N., Ju T.}:
\newblock A general and efficient method for finding cycles in {{3D}} curve
  networks.
\newblock \emph{ACM Transactions on Graphics 32}, 6 (Nov. 2013), 1--10.
\newblock \href {https://doi.org/10.1145/2508363.2508423}
  {\path{doi:10.1145/2508363.2508423}}.

\end{thebibliography}
\clearpage
\appendix
\section{Proof of Lemma \ref{lem:main}}
\label{app:lemma_proof}
\new{Intuitively, in $\mathbb{R}^3$, a ray and a surface intersect \emph{transversally} if at each intersection point the ray is not tangent to the surface. At such intersection, the ray's tangent and the surface's tangent plane span the whole $\mathbb{R}^3$. More formally,}

\label{app:proof}

\begin{definition} \cite{guilleminDifferentialTopology2010}.
	The smooth map $f$ is said to be \textit{transversal} to the submanifold $Z$, if at every point $x$ in the preimage of $Z$, where $f(x)=y$ and $T$ is a tangent space,
	\[
	\mathrm{Image}(df_x) + T_y(Z) = T_y(Y).
	\]
\end{definition}

\begin{figure}[H]
	\includegraphics[width=\linewidth]{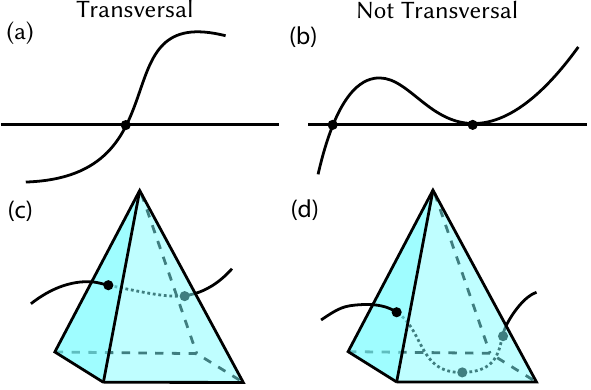}
	\caption{Two curves in 2D are transversal (a) if for every intersection their tangents at the intersection spans $\mathbb{R}^2$. A curve and a surface in 3D are transversal (c) if for every intersection the curve's tangent and the surface's tangent plane spans $\mathbb{R}^3$. Two submanifolds are not transversal if their tangent spaces do not span the ambient space (b,d).}
	\label{fig:transversal}
\end{figure}

In our context, the \emph{orientation number} at an intersection point measures whether the ray and the normal to the surface are pointing in the same direction, i.e., it is $\mathrm{Sign}(n \cdot \overrightarrow{pq})$. For a more formal definition, please see \cite{guilleminDifferentialTopology2010}.

\new{In $Y = \mathbb{R}^3$, we first examine the case of a smooth surface $Z$ without boundary. The smooth map $f: X \rightarrow Y$ is a parameterization of a ray $f: \mathbb{R}^+ \rightarrow \mathbb{R}^3$. More generally,}

\begin{proposition} \cite{guilleminDifferentialTopology2010}
If $f: X \rightarrow Y$ is transversal to $Z$, then $f^{-1}(Z)$ is a finite number of points, each with an orientation number $\pm 1$. We define the \emph{intersection number} $I(f,Z)$ to be the sum of these orientation numbers. 
\end{proposition}

\new{The main paper refers to the intersection number as the `number of signed intersections' as it is customary in computer graphics \cite{fengWindingNumbersDiscrete2023}.}

For boundaryless $X$ and $Z$, when at least one of them is compact, \cite{guilleminDifferentialTopology2010} formulate the following proposition:

\begin{proposition} \cite{guilleminDifferentialTopology2010}
	Homotopic maps have the same intersection numbers.
	\label{prop:base}
\end{proposition}
 
\new{We will use this result to show that for a smoothly rotating ray, forming a homotopy, the intersection number can only change when the rotating ray passes a boundary of a surface. More generally,} we formulate our result:
 
\begin{proposition} \new{Let $X$ be an oriented manifold without boundary, and $M$ be a oriented submanifold with boundary of $Y$, $\text{dim}\; X + \text{dim} \; M = \text{dim}\; Y$. Let $f_0,f_1$ be smooth maps $X \rightarrow Y$, also $F$ be a homotopy between them, i.e., $F: [0,1] \times X \rightarrow Y, F(0) = f_0, F(1)=f_1$, and $f_0, f_1$ intersect $M$ transversally and for each $s, F(s)$ intersects the boundary $\partial M$ not more than once.
Then the oriented intersection number $I(F,M)$ only changes when $F(s,*)$ intersects the boundary $\partial M$, and only by increments of 1.}
\end{proposition}

%For our case in $\mathbb{R}^3$, we can think of $F(s,X)$ as smoothly moving our ray $\overrightarrow{pq}$ within an open region $A_i$ (Lemma \ref{lem:main}) on the surface of the unit sphere $S_p^2$.
 \begin{proof}. \new{We denote as $\overline{M}$ an arbitrary closure of $M$ such that the smooth manifold $Z = M \cup \overline{M}$ and $M \cap \overline{M} = \emptyset$. A transversal intersection on $Z$ may belong to $M$ or $\overline{M}$, thus, the intersection number $I(F, Z) = I(F,M) + I(F,\overline{M})$. By proposition~\ref{prop:base}, $I(F,Z)$ does not change for those values of $s$. }
 	
\new{Suppose that $F(s,X)$ intersects the boundary of $M$ only at $s=s_0$, maximum at one point. For this point, as the boundary is crossed to $s_0+\varepsilon$, the number of intersection points is no longer one, but still finite due to transversality, so $\overline{M}$ gains (or loses) this intersection, and since $I(F, Z)$ is constant, $I(F, M)$ loses (or gains) an intersection (Fig.~\ref{fig:proof_intuition}).  Therefore, the only value of $s$ where $I(F,M)$ can change is $s=s_0$ and the intersection number can only change by 1. Q.E.D.}
\end{proof}
\begin{figure}[H]
	\includegraphics[width=\linewidth]{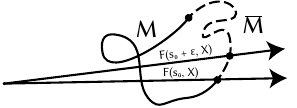}
	\caption{The intersection number $I(F,Z)$ remains unchanged. When $Z$ is decomposed into a manifold $M$ with boundary and a closure $\overline{M}$ the intersection number $I(F,M)$ only changes at the boundary. }
	\label{fig:proof_intuition}
\end{figure}

\new{The only technicality is that in the propositions we assumed boundaryless $X$, while $\mathbb{R}^+$ has a boundary $\{0\}$. This is not an actual limitation: our homotopy $F(s,X)$ never changes the value $F(s,0)$, so the proposition still stands.}

\section{Boundary Element Method Details}
\label{app:bem_details}

We follow the same notations and general setup as in Sec.~\ref{sec:ray_int}. 

For the flat 2D Laplace's equation, the fundamental solution is 
\begin{equation}
	G(\xi,\eta) = -\frac{1}{2\pi}\mathrm{ln} | |\xi-\eta| |.
	\label{eq:fsol_laplace}
\end{equation}

To compute the sign of an intersection, we look at the sign of the dot product between the normal of the surface at the intersection point, computed via differentiating Eq.~(\ref{eq:representation}), and the ray direction, e.g.: 

\begin{equation}
	\frac{\partial f}{\partial u} = \int _{\eta \in \Gamma} \frac{\partial G(\xi,\eta)}{\partial u} \frac{\partial f(\eta)}{\partial n} - \frac{\partial^2 G (\xi,\eta)}{\partial n \partial u} f(\eta) d\Gamma,
	\label{eq:drepresentation}
\end{equation}

where  $\frac{\partial G(\xi, \eta)}{\partial \xi}$ and $\frac{\partial^2 G (\xi,\eta)}{\partial n \partial \xi}$ are easily calculated in closed form 
%(see Eqs.~\ref{eq:fsol_laplace},\ref{eq:fsol_biharmonic}).
(see Eq.~\ref{eq:fsol_laplace}).

\section{Sum-of-Squares Formulations}
\label{app:sos}
\subsection{Intersections of B\'ezier curve projections}
Here we outline the formulation for finding an intersection of two degree two B\'ezier curves, when projected onto a unit sphere via Sum-of-Squares (SOS) relaxation \cite{SOS}. This formulation can be easily adapted to degree 3 B\'ezier curves with a simple change of basis functions. We are following the notations of \cite{SOS2}.

In general, the problem of finding an intersection of spherical projections of two polynomial curves $f_1(t_1), f_2(t_2): [0,1] \rightarrow \mathbb{R}^3$ can be formulated as a following low-order polynomial optimization:

\begin{align}
	\begin{split}
 &\min _{t_1, t_2, R} t_1^2\\
 &f_1(t_1) R = f_2(t_2)\\
 &t_1, t_2 \in [0,1], R>0
 \end{split}
 \label{eq:bezier_sphere_int}
\end{align}

Concretely, for B\'ezier curves, basis functions are:
\begin{align}
	\begin{split}
		&\phi_1(t) := (1-t)^2\\
		&\phi_2(t) := 2t(1-t)\\
		&\phi_3(t) := t^2.
	\end{split}
\end{align}

Then we can define two B\'ezier curves as $f_1,f_2: [0,1] \rightarrow \mathbb{R}^3$, where $F_i, G_i \in \mathbb{R}^3$ are the control vertices, i.e., $f_1(t_1) = \sum_{i=1}^3 F_i \phi_i(t_1)$ and $f_2(t_2) = \sum_{i=1}^3 G_i \phi_i(t_2)$. Then Eq.~\ref{eq:bezier_sphere_int} has the following SOS relaxation:

\begin{align}
	\begin{split}
	&\gamma* = \left\{  \begin{array}{l}  \max_{\gamma,s_g,p_h} \gamma \\
		\textrm{s.t.} \;
		t_1^2 - \gamma - \sum_{g \in \mathcal{G}} s_g g - \sum_{h\in \mathcal{H}} p_h h \in \Sigma\\
		s_g \in \Sigma_{d_1}\\
		p_h \in \mathbb{R}[t]_{d_2}\\
	\end{array}  \right\} \\
    &\mathcal{G} = \{(1-t_1)t_1, (1-t_2)t_2, R\}\\
    &\mathcal{H} = f_1(t_1) R - f_2(t_2),
    \end{split}
\end{align}

The objective function $ t_1^2$ is only needed to pick out a particular intersection. Once an intersection is found with values $t'_1, t'_2$, we repeat the process after subdividing the curve at $t'_1$ and taking its $t_1 > t'_1$ chunk. 

Here $\Sigma_{d_1}$ are SOS polynomials of maximum degree $d_1$, and $\mathbb{R}[t]_{d_2}$ is a subset of polynomials of maximum degree $d_2$. This is a canonical SOS formulation, which can be solved by standard SDP solvers. In our implementation, we use YALMIP \cite{Lofberg2004,Lofberg2009}. We use $d_1=4, d_2=2$ for both quadratic and cubic B\'ezier curves.

Self-intersections of a spherical projection of a B\'ezier curve are implemented in exactly the same way by substituting $f_2=f_1$ and replacing $-t_1^2$ with $+(t_1-t_2)^2$. In this case, we ask to maximize the difference between $t_1$ and $t_2$. For this problem, solver does not indicate that a solution is not found if there is no intersection, so we need to check if the returned $t_1, t_2$ are sufficiently different; in our implementation it is $|t_1-t_2|>10^{-6}$.

\subsection{Ray Intersection with a B\'ezier Triangle}
Similarly, we can consider the problem of finding the first intersection of a ray with a cubic B\'ezier triangle. As done in \cite{SOS}, this can be trivially generalized to other degrees of B\'ezier triangles or to an arbitrary NURBS patch. 

A cubic B\'ezier triangle has 10 control vertices $F_i \in \mathbb{R}^3$ and is expressed as $f(u,v) = \sum_i \phi_i F_i$, where the basis functions $\phi_i$ are:

\begin{align}
	\begin{split}
		&\phi_1(u,v) = -(u + v - 1)^3\\
		&\phi_3(u,v) =  3u(u + v - 1)^2\\
		&\phi_5(u,v) = -6uv(u + v - 1)\\
		&\phi_7(u,v) = v^3\\
		&\phi_9(u,v) = 3u^2 v
	\end{split}
	\begin{split}
		&\phi_2(u,v) = 3v(u + v - 1)^2\\
		&\phi_4(u,v) = -3v^2 (u + v - 1)\\
		&\phi_6(u,v) = -3u^2 (u + v - 1)\\
		&\phi_8(u,v) = 3u v^2 \\
		&\phi_{10}(u,v) = u^3.
	\end{split}
\end{align}

The polynomial problem in Eq.~(\ref{eq:ray_bezier_int}) in the main document has the following SOS relaxation:

\begin{align}
	\begin{split}
		&\gamma* = \left\{  \begin{array}{l}  \max_{\gamma,s_g,p_h} \gamma \\
			\textrm{s.t.} \;
			t^2 - \gamma - \sum_{g \in \mathcal{G}} s_g g - \sum_{h\in \mathcal{H}} p_h h \in \Sigma\\
			s_g \in \Sigma_{d_1}\\
			p_h \in \mathbb{R}[t]_{d_2}\\
		\end{array}  \right\} \\
		&\mathcal{G} = \{u,v,1-u,1-v,t\}\\
		&\mathcal{H} = f(u,v) - O-tR.
	\end{split}
\end{align}

We found that low degrees of $d_1 = 4, d_2 = 2$ are sufficient for exact recovery.
%\bibliographystyle{eg-alpha-doi} 
%\bibliography{wn}   
\end{document}